\documentclass[a4paper,11pt]{article}
\pdfoutput=1

\usepackage{jheppub}
\usepackage{tikz}
\usetikzlibrary{shapes}
\usetikzlibrary{plotmarks}
\usepackage{tikz-3dplot}

\newcommand{\ev}[1]{\langle #1\rangle}
\DeclareMathOperator{\Map}{Map}

\title{From 5d Flat Connections to 4d Fluxes\\ (the Art of Slicing the Cone)}
\author{Jim Lundin,}
\author[1]{Roman Mauch}
\author[2]{and Lorenzo Ruggeri}
\affiliation[1]{Department of Physics and Astronomy, Uppsala Universitet, 752 37 Uppsala, Sweden}
\affiliation[2]{Yau Mathematical Sciences Center, Tsinghua University, Beijing, 100084, China}
\emailAdd{jimeriklundin@gmail.com}
\emailAdd{roman.mauch@physics.uu.se}
\emailAdd{ruggeri@mail.tsinghua.edu.cn}
\preprint{UUITP-12/23}

\abstract{We compute the Coulomb branch partition function of the 4d $\mathcal{N}=2$ vector multiplet on closed simply-connected quasi-toric manifolds $B$. This includes a large class of theories, localising to either instantons or anti-instantons at the torus fixed points (including Donaldson-Witten and Pestun-like theories as examples). The main difficulty is to obtain flux contributions from the localisation procedure. We achieve this by taking a detour via the 5d $\mathcal{N}=1$ vector multiplet on closed simply-connected toric Sasaki-manifolds $M$ which are principal $S^1$-bundles over $B$. The perturbative partition function can be expressed as a product over slices of the toric cone. By taking finite quotients $M/\mathbb{Z}_h$ along the $S^1$, the locus picks up non-trivial flat connections which, in the limit $h\to\infty$, provide the sought-after fluxes on $B$. We compute the one-loop partition functions around each topological sector on $M/\mathbb{Z}_h$ and $B$ explicitly, and then factorise them into contributions from the torus fixed points. This enables us to also write down the conjectured instanton part of the partition function on $B$.}

\addtocontents{toc}{\protect\setcounter{tocdepth}{2}}

\begin{document}

\maketitle
\flushbottom

\section{Introduction}

    Four-dimensional $\mathcal{N}=2$ supersymmetric quantum field theories (SQFTs) have proven, over the years, to be an excellent laboratory for studying connections between QFT and mathematics and for enhancing our insight into more realistic QFTs. The work of Witten \cite{Witten:1988ze} showed how to interpret Donaldson invariants of four-manifolds \cite{Donaldson:1985zz} in terms of observables in a topological subsector of the full $\mathcal{N}=2$ SQFT counting instantons. Later, the derivation of the Seiberg-Witten exact prepotential of an $SU(2)$ $\mathcal{N}=2$ gauge theory in the IR \cite{Seiberg:1994rs} greatly improved our understanding of strongly coupled QFTs. This result relied on a version of electric-magnetic duality and on the holomorphicity of the moduli space of vacua. It was later reinterpreted by Nekrasov \cite{Nekrasov:2002qd}, as the non-equivariant limit of the equivariant volume of the instanton moduli space in the $\Omega$-background. This computation employed supersymmetric localisation (see \cite{Pestun:2016zxk} for a comprehensive review) to reduce the integral over the entire instanton moduli space to a discrete set of fixed points under a $T^2$-isometry. The Nekrasov partition function is the building block of Pestun's result \cite{Pestun:2007rz} for the exact partition function of the $\mathcal{N}=2^\ast$ theory on $S^4$: on this background, the field strength localises to point-like instantons at the north pole and point-like anti-instantons at the south pole. The partition function is obtained by gluing two copies of the Nekrasov partition function, one for instantons and one for anti-instantons. Subsequently, a plethora of exact results for SQFTs have been obtained. Among these, related to our work are studies of $\mathcal{N}=2$ topologically twisted theories on compact toric manifolds \cite{Bawane:2014uka,Rodriguez-Gomez:2014eza,Bershtein:2015xfa,Bershtein:2016mxz,Bonelli:2020xps}. See also \cite{Hosseini:2018uzp,Closset:2022vjj} for topologically twisted theories on $M_4\times S^1$, where $M_4$ is a compact 4d toric Kähler manifold.

    More recently, in \cite{Festuccia:2018rew,Festuccia:2019akm,Festuccia:2020yff} a framework was developed to generalise Pestun's result to any compact simply-connected four-manifold with a $T^2$-isometry with isolated fixed points and an arbitrary distribution of (anti-)instantons at those points. The partition function is conjecturally obtained by gluing (anti-)instanton Nekrasov partition functions. However, this is not the full story, as four-manifolds $B$ with non-trivial $H^2(B;\mathbb{Z})$ also admit gauge configurations where the field strength has flux on non-contractible cycles. The way to see this is as follows: for non-Abelian gauge groups $G$, in order to solve the localisation equations, one chooses a ``diagonal'' gauge such that scalars $\tilde{\phi},\varphi$ in the locus take values in the Cartan subalgebra. However, even for the trivial gauge bundle, for $H^2(B;\mathbb{Z})\neq0$, the maps $\tilde\phi,\varphi\in\Map(B,\mathfrak{g})$ might be topologically non-trivial, in which case there is an obstruction to choosing diagonal gauge. Such a choice is only possible if the Cartan component of the gauge field is a connection of a non-trivial torus-bundle \cite{Blau:1994rk}. It was conjectured in \cite{Nekrasov:2003vi,Festuccia:2018rew} that such flux contributions enter as a shift of the Coulomb branch parameter in each copy of the Nekrasov partition function. The full partition function is then obtained as a sum over these flux sectors. In the literature \cite{Nekrasov:2003vi,Nakajima:2003pg,Nakajima:2003uh,Gottsche:2006tn,Bawane:2014uka,Bershtein:2015xfa,Bershtein:2016mxz,Hosseini:2018uzp,Crichigno:2018adf,Festuccia:2018rew,Festuccia:2019akm,Bonelli:2020xps}, the sum is over equivariant fluxes which are then constrained by stability equations.
     
    The main objective of this work is to compute the one-loop partition function around fluxes, for the 4d $\mathcal{N}=2$ vector multiplet on compact simply-connected quasi-toric manifolds. In particular, we allow the theories to localise to different distributions of (anti-)instantons at the torus fixed points (also known as \textit{Pestunization}).
    We explicitly compute the one-loop partition function of these theories around each flux sector, whose factorised form\footnote{By ``factorised'' we mean an expression for the partition function in terms of a product over its contributions from individual torus fixed points on the manifold. This notion was introduced in \cite{Pasquetti:2011fj,Beem:2012mb}.} agrees with the form conjectured in \cite{Nekrasov:2003vi,Festuccia:2018rew}. From this form we can read out the shifts of the Coulomb branch parameter at each torus fixed point by the fluxes and write down the Coulomb branch partition function\footnote{This is the holomorphic ``integrand'' of the full partition function and we make this distinction since we will not concern ourselves with the remaining integral over the Coulomb branch parameter in this work.}, including instanton and classical parts. 
    
    The main difficulty is obtaining flux solutions from the 4d localisation equations. We circumvent this issue by taking a detour via the 5d $\mathcal{N}=1$ vector multiplet on closed toric Sasakian manifolds. This setup is well-studied \cite{Kallen:2012cs,Kallen:2012va,Imamura:2012efi,Lockhart:2012vp,Kim:2012qf,Qiu:2013pta,Qiu:2013aga,Qiu:2014oqa,Alday:2015lta,Festuccia:2016gul}, in particular for $S^5$ and $Y^{p,q}$ spaces (see also \cite{Qiu:2016dyj} for a review). In order to, ultimately, return to 4d, we assume the five-dimensional manifold $M$ to be a (non-trivial) principal $S^1$-bundle
    \begin{equation}\label{eq--1.fibration}
        \begin{tikzcd}[column sep=small]
            S^1\ar[r,hook] & M\ar[r,"p"] & B
        \end{tikzcd}
    \end{equation}
    over $B$. We can usually have multiple different such bundles over $B$, which lead to different (anti-)instanton distributions (depending on the relative orientation of the free $S^1$-direction with respect to the Reeb vector). For simply-connected $M$, the perturbative partition function was computed in \cite{Qiu:2014oqa} and is given by
    \begin{equation*}
        Z_M=\prod_{\alpha\in\Delta}\prod_{\Vec{n}\in \mathcal{C}\cap\mathbb{Z}^3}(\vec{n}\cdot\vec{\R}+ \ii\alpha(a) )\prod_{\Vec{n}\in \mathring{\mathcal{C}}\cap\mathbb{Z}^3}(\vec{n}\cdot\vec{\R}- \ii\alpha(a) ).
    \end{equation*}
    with $\vec{\R}$ the Reeb vector and $\mathcal{C}$ the three-dimensional moment map cone of the $T^3$-action on $M$. We first rewrite this result as a product over two-dimensional slices $\mathcal{B}_t$ of the cone along the free $S^1$-direction (as displayed in \autoref{fig--3.2.slices}), labelled by their integral charge $t$ under the $S^1$: 
    \begin{equation*}
        Z_M= \prod_{\alpha\in\Delta}\prod_
        t\prod_{(n_1,n_2)\in\mathcal{B}_t}(n_1\epsilon_1+n_2\epsilon_2+t\tfrac{\R_3}{l_3}+ \ii\alpha(a) )\prod_{(n_1,n_2)\in\mathring{\mathcal{B}}_t}(n_1\epsilon_1+n_2\epsilon_2+t\tfrac{\R_3}{l_3}- \ii\alpha(a) ).
    \end{equation*}
    Here, $\epsilon_1,\epsilon_2$ are the equivariance parameters corresponding to the remaining $T^2$-action and $l_3$ some constant.
    
    As a first step towards $B$, we introduce a quotient of the simply-connected $M$ along the free $S^1$ by a finite subgroup, $X=M/\mathbb{Z}_h$. The quotient has $\pi_1(X)=\mathbb{Z}_h$ and the localisation locus must be extended to include topologically non-trivial flat connections, valued in the Cartan of the gauge algebra. The effect of the quotient on the partition function is to discard most slices and only keep the ones for which $t=c_1\mmod h$, controlled by the first Chern class $c_1$ of the corresponding $U(1)$-bundles\footnote{Here and below, when we write $c_1\mmod h$ we are implicitly using the isomorphism $H^2(X)\simeq\mathbb{Z}_h\oplus\mathbb{Z}^{b_2(X)}$ under which $c_1$ takes values in $\mathbb{Z}_h$ (see \autoref{app-topology}).}. The partition function on $X$ is then given by a sum over these topological sectors. We suspect that this procedure can be interpreted as gauging a $\mathbb{Z}_h$ 0-form symmetry on $M$.
    
    Finally, we return to $B$ by taking the limit of large $h$. The non-trivial connections in the 5d locus on $X$, upon reduction, become connections in the 4d locus on $B$ with flux on a two-cycle. In terms of the partition function, instead of a product over different slices, for large $h$ only a single slice of the cone survives, labelled by the aforementioned flux $\mathfrak{m}$:
    \begin{equation*}   
        Z_B= \prod_{\alpha\in\Delta}\prod_{(n_1,n_2)\in\mathcal{B}_\mathfrak{m}}(n_1\epsilon_1+n_2\epsilon_2+\alpha(\mathfrak{m})\tfrac{\R_3}{l_3}+\ii\alpha(a))\prod_{(n_1,n_2)\in\mathring{\mathcal{B}}_\mathfrak{m}}(n_1\epsilon_1+n_2\epsilon_2+\alpha(\mathfrak{m})\tfrac{\R_3}{l_3}-\ii\alpha(a)).
    \end{equation*} 
    Note that one can think of $Z_M$ as counting holomorphic functions\footnote{Since $C(M)$ is a toric manifold, the integral lattice of the toric cone $\mathcal{C}$ is in 1:1 correspondence to holomorphic functions on $C(M)$. See \cite{Qiu:2016dyj} for precise relations to the one-loop calculation.} on the metric cone $C(M)$. Then $Z_B$ counts precisely the holomorphic functions on $C(M)$ that have charge $c_1$ under rotations along the $S^1$-fibre. 
    For the full partition function, the topological sectors are summed over as before:
    \begin{equation*}
        \mathcal{Z}_B=\sum_{\mathfrak{m}}\int_{\mathfrak{h}}d a\;e^{-S_{\text{cl}}}\cdot Z_B\cdot Z_B^{\text{inst}}.
    \end{equation*}    
    The one-loop partition function $Z_B$ for flux $\mathfrak{m}$ can then be factorised using 5d techniques. Note that, depending on which fibration \eqref{eq--1.fibration} we choose, the remaining slice $\mathcal{B}_\mathfrak{m}$ can be either compact or non-compact. The theory on $B$ corresponding to the former is commonly known as the topologically twisted theory, localising to instantons on all torus fixed points (resulting in an elliptic complex; e.g. Donaldson-Witten theory). The latter case appears for more general distributions of (anti-)instantons (resulting in a transversally elliptic complex; e.g. Pestun's theory on $S^4$).

    The 5d detour described above makes the inclusion of fluxes into the 4d locus a simple consequence of the non-trivial nature of the fibration \eqref{eq--1.fibration}. Given that we obtain flux contributions on $B$ as the base of an $S^1$-fibration, one might wonder whether a similar technique could be applied to analyse non-Abelian features of the locus from non-Abelian fibrations over the four-manifold.

    Finally, as a disclaimer, whenever $H^2(M;\mathbb{Z})\neq0$, our procedure might only produce part of the sum over all possible fluxes in the full partition function on $B$. This is because our starting point in 5d is a localisation result \cite{Qiu:2014oqa}, whose locus does not include any flux. As a consequence, our result for the 4d partition function might miss out on these fluxes. In order to include them, it would be necessary to allow for a more general 5d localisation locus (possibly by relaxing reality conditons or including singular configurations, cf. \cite{Benini:2015noa,Closset:2015rna} in 3d and 2d).

    The outline is as follows. In \autoref{sec--2} we introduce some basic facts about toric Sasakian geometry and the condition under which these manifolds admit a free $S^1$-action. As an example, we consider $Y^{p,q}$. Then, in \autoref{sec--3}, we explain the supersymmetric setup on these manifolds and on their quotients. We also discuss how flat connections in $X$ lead to flux on $B$. In \autoref{sec--4} we perform in detail the steps of slicing and reducing outlined above and compute the one-loop partition function around flat connections on the quotient $X$ and around fluxes on the four-dimensional base $B$. In particular, we show how each flux sector arises as a single slice of the three-dimensional cone. The construction is illustrated by explicit examples in \autoref{sec--5}, considering $S^5$, $Y^{p,q}$, $T^{1,1}$ and $A^{p,q}$. The factorisation properties of the partition functions are studied in \autoref{sec--6} where, employing the shifts derived from the one-loop part around fluxes, we can finally write the full partition function, including instanton and classical contributions. We summarise our results in \autoref{sec--7} and comment on possible future directions. 

    This work is a generalisation of earlier work \cite{Lundin:2021zeb} on $\mathbb{CP}^2$ by two of the authors.

\section{Toric Sasaki-Manifolds}\label{sec--2}

    In this section, we briefly review some facts about toric Sasaki-manifolds on which we later place our theory. A detailed account of contact structures and Sasakian geometry can be found in \cite{Boyer:2008era}, for toric geometry see \cite{Fulton:1993}. Throughout the article, we assume all manifolds as being smooth and connected.

    \subsection{Toric Sasakian Geometry}
    
        \paragraph{Contact Structures.} 
            The basic underlying structure of Sasakian manifolds is a contact structure. It can be viewed as 
            the odd-dimensional cousin of a symplectic structure on even-dimensional manifolds. More precisely, on a $(2n-1)$-dimensional manifold $M$, consider a field $\xi\subset TM$ of hyperplanes on $M$. This field can be expressed as $\xi=\ker\kappa$, where $\kappa$ is a one-form\footnote{In fact, $\kappa$ is globally defined only if $\xi$ is co-orientable, i.e. the line bundle $TM/\xi$ is trivial. Equivalently, one can always find a global $\kappa$ with $\xi=\ker\kappa$ if both $M$ and $\xi$ are orientable.} on $M$. If $\kappa\wedge(\dd\kappa)^{n-1}\neq0$ everywhere\footnote{An equivalent formulation of this condition is $(\dd\kappa)^{n-1}|_{\xi}\neq0$. Hence, we anticipate that $\dd\kappa$ provides a symplectic structure for a suitable codimension-one submanifold.} on $M$, then $\xi$ is called a \textit{contact structure}, $\kappa$ a \textit{contact form} and $(M,\xi)$ a \textit{contact manifold}. Note that in this case $\kappa\wedge(\dd\kappa)^{n-1}$ defines a volume form of $M$ (i.e., in particular, $M$ must be orientable).  
            
            To a contact form $\kappa$ we can associate a vector field $\R$ defined by the equations 
            \begin{equation}\label{eq--2.1.reeb}
                \iota_{\SR}\kappa=1,\qquad\iota_{\SR}\dd\kappa=0.
            \end{equation}
            This is called the \textit{Reeb vector field}. Since we consider $M$ to be Riemannian, we demand some compatibility conditions between the metric and contact structure: the metric $g$ should be preserved, i.e. $\mathcal{L}_{\SR} g=0$ and there should exist an almost complex structure $J$ on $\ker\kappa$ which can be extended to $TM$ via $J(\R)=0$ such that for vector fields $X,Y$ on $M$, $g(X,Y)=\frac{1}{2}\dd\kappa(X,J(Y))+\kappa(X)\kappa(Y)$ (note that $(\ker\kappa,J|_{\ker\kappa})$ provides $M$ with an almost CR-structure). If these conditions are satisfied, $(M,g,\kappa,J)$ is called a \textit{K-contact manifold}. Note that a consequence of the second condition and \eqref{eq--2.1.reeb} is $\kappa=g(\R,\us)$.
            
            The Reeb vector field generates a flow on $M$ and its orbits (which are geodesics) can be viewed as the leaves of a one-dimensional foliation, the so-called \textit{Reeb foliation} $\mathcal{F}_{\SR}$. Note that the tangent bundle $TM$ can be orthogonally decomposed as $TM\simeq\ker\kappa\oplus L_{\SR}$ with respect to $g$, where $L_{\SR}$ is the (trivial) line bundle consisting of vectors tangent to the leaves of $\mathcal{F}_{\SR}$. 
            
            One way to characterise $M$ is by the regularity property of the Reeb foliation. $\mathcal{F}_{\SR}$ is called \textit{quasi-regular} if there is an integer $k$ such that for every point $x\in M$ there is a foliated chart through which each leaf passes at most $k$-times. $\mathcal{F}_{\SR}$ is called \textit{regular} if $k=1$ and \textit{irregular} if it is not quasi-regular\footnote{Note that in the case of a regular (quasi-regular, irregular) foliation we also talk about a regular (quasi-regular, irregular) Reeb vector field or contact form. However, different contact forms in a contact \textit{structure} can produce drastically different foliations, so it does not make sense to talk about a regular (quasi-regular, irregular) contact structure.}. Note that if $M$ is compact and $\mathcal{F}_{\SR}$ quasi-regular, the orbits are circles and we get a locally free $S^1$-action on $M$ (which is isometric since ${\R}$ is Killing). If $\mathcal{F}_{\SR}$ is regular, the $S^1$-action is free. Hence, in this case, we have a principal $S^1$-bundle $\pi:M\rightarrow B$ over the space of leaves $B=M/\mathcal{F}_{\SR}$ which is a compact, symplectic manifold with symplectic form $\omega$ such that $\pi^\ast\omega=\dd\kappa$. In the locally free case instead, $B$ is a symplectic orbifold. In contrast, in the irregular case, ${\R}$ generates a group action with orbits whose closure is isomorphic to a torus.
        
        \paragraph{Sasakian Manifolds.} 
            As mentioned above, a Sasakian manifold is a special case of a K-contact manifold and can be viewed as the
            odd-dimensional cousin of a Kähler manifold in even dimensions. More precisely, a K-contact manifold $(M,g)$ is Sasakian if its metric cone $(C(M),\bar{g})$ with $C(M)=\mathbb{R}_{>0}\times M$ and $\bar{g}=\dd r^2+r^2g$ is Kähler. The corresponding Kähler form and complex structure are
            \begin{equation}\label{eq--2.1.kahlerstructure}
                \bar{\omega}=\frac{1}{2}\dd(r^2\kappa),\qquad \bar{J}=J+\frac{1}{r}\zeta\otimes\dd r-r\partial_r\otimes\kappa.
            \end{equation}
            Here, $J$ and $\kappa$ are trivially extended to $C(M)$ and $\zeta=\bar{J}(r\partial_r)$ is the characteristic vector field. The Kähler condition is equivalent to $\bar{J}$ being covariantly constant, $\bar{\nabla}\bar{J}=0$, with respect to the Levi-Civita connection (see e.g. \cite{McDuff:2017}, lemma 4.2.5). It follows that $\|\zeta\|=r^2$, $\mathcal{L}_\zeta\bar{g}=0$ and $\zeta$ is tangent to hypersurfaces of constant $r$. Hence, $M$ can be naturally embedded into $C(M)$ as the hypersurface with $r=1$ and $\zeta$ can be identified with the Reeb vector field of $M$.
            
            As described above, we can characterise a Sasakian manifold as regular, quasi-regular, or irregular, depending on its Reeb foliation. For a compact regular (quasi-regular) Sasakian manifold, the space of leaves $B$ is a compact Kähler manifold (orbifold) (see \cite{Boyer:2008era}, theorem 7.1.3). In turn, $\mathcal{F}_{\SR}$ is a transverse Kähler foliation. Note that regular Sasakian manifolds are rare (e.g. $S^5$ with $B=\mathbb{CP}^2$) and most Sasakian manifolds are either quasi-regular or irregular (e.g. $Y^{p,q}$ which we discuss in more detail below). Since the aim of this work is to start from a five-dimensional Sasakian manifold and reduce, along a free direction, to four dimensions, it is clear that this direction, in general, has to be different from the Reeb vector field.
    
        \paragraph{Toric Sasakian Geometry.} 
            A $(2n-1)$-dimensional Sasakian manifold $M$ is called \textit{toric} if, on its cone $(C(M),\bar{g},\bar{\omega},\bar{J})$, there exists an effective, holomorphic $T^{n}$-action that is Hamiltonian and such that $\zeta\in\mathfrak{t}_{n}$ is an element\footnote{Note the abuse of notation by identifying an element in $\mathfrak{t}_{n}$ with its induced vector field.} of the Lie algebra of $T^{n}$. The Hamiltonian property implies existence of a $T^{n}$-equivariant moment map 
            \begin{equation}\label{eq--2.1.momentmap}
                \mu:C(M)\longrightarrow \mathfrak{t}_{n}^{\vee}\simeq\mathbb{R}^{n},
            \end{equation}
            where $\mathfrak{t}_{n}^\vee$ denotes the dual, such that
            \begin{equation}
                \forall v\in\mathfrak{t}_{n}:\;\dd\langle\mu,v\rangle=\iota_v\bar{\omega}.
            \end{equation}
            From \eqref{eq--2.1.kahlerstructure} it follows that (up to an additive constant) the moment map is determined by
            \begin{equation}\label{eq--2.1.hamiltonian}
                \langle\mu,v\rangle=\frac{1}{2}r^2\kappa(v).    
            \end{equation}
            Moreover, the image $\mu(C(M))\cup\{0\}$ is a \textit{strictly convex}, \textit{rational}, \textit{polyhedral cone} $\mathcal{C}$ in $\mathfrak{t}_{n}^\vee\simeq\mathbb{R}^{n}$ (see \cite{Lerman:2001zua}, theorem 4.2), i.e. $\mathcal{C}$ can be presented as
            \begin{equation}\label{eq--2.1.cone}
                \mathcal{C}=\{u\in\mathfrak{t}_{n}^\vee|\forall i=1,\dots,m:\langle u,v_i\rangle\ge0\}.
            \end{equation}
            Here, $v_i\in\mathfrak{t}_{n}$, $i=1,\dots,m$ are the ``inward-pointing normals'' of the facets (codimension-one faces) of the cone and $m$ the number of facets. The rationality means that the $v_i$ are elements of the lattice of circle subgroups of $T^{n}$, i.e. upon $\mathfrak{t}_{n}\simeq\mathbb{R}^{n}$ we have $v_i\in\mathbb{Z}^{n}$. Furthermore, we can assume $\{v_i\}_{i=1,\dots,m}$ to be minimal and primitive. Since the cone $\mathcal{C}$ is strictly convex, its base is a compact convex $(n-1)$-dimensional polytope; hence, $m\ge n$. An intuitive way of thinking about the moment map \eqref{eq--2.1.momentmap} is as a $T^{n}$-fibration over $\mathcal{C}$ where at the $i$th facet of the cone the circle subgroup specified by $v_i$ degenerates. We can also define the \textit{dual cone} $\mathcal{C}^\vee\subset\mathfrak{t}_{n}$ as
            \begin{equation}\label{eq--2.1.dualcone}
                \mathcal{C}^\vee=\{v\in\mathfrak{t}_{n}|\forall u\in\mathcal{C}^\vee:\langle u,v\rangle\ge0\},    
            \end{equation}
            which is again a convex, rational, polyhedral cone \cite{Sparks:2010sn}. Note that, for the characteristic vector field, $\zeta\in\mathring{\mathcal{C}}^\vee$ since $\langle\mu,\zeta\rangle=\frac{1}{2}r^2>0$ (here and in the following, $\mathring{\mathcal{C}}$ denotes the \textit{interior} of $\mathcal{C}$). Note that, once we pick a basis $\{e_i\}\subset T^{n}$, we are free to perform $SL(n,\mathbb{Z})$-transformations which in turn generate $SL(n,\mathbb{Z})$-transformations of $\im\mu$. Hence, the cones and polytopes are unique only up to such transformations.
            
            Recall that the toric Sasakian manifold $M$ is naturally embedded into $C(M)$ as $\{1\}\times M$. Hence, we can define a moment map on $M$ by restriction $\mu|_{\{1\}\times M}$ (which we again denote by $\mu$) and, using \eqref{eq--2.1.hamiltonian}, the image of $M$ under the moment map is given by
            \begin{equation}\label{eq--2.1.momentpolytope}
                \mu(M)=\{u\in\mathcal{C}|\langle u,\R\rangle=\tfrac{1}{2}\},
            \end{equation}
            i.e. the intersection of $\mathcal{C}$ with the \textit{characteristic plane} $\{u\in\mathfrak{t}^\vee_{n}|\langle u,\R\rangle=\frac{1}{2}\}$. This image is an $(n-1)$-dimensional compact convex polytope \cite{Boyer:2000}. Similarly to $C(M)$, we can view $M$ as a $T^{n}$-fibration where the base is now $\mu(M)$, i.e. $(n-1)$-dimensional. This is illustrated in \autoref{fig--2.1.cone} for the five-dimensional case ($n=3$) which is the relevant one for us. In particular, at the vertices of the polytope, only an $S^1$ survives whose orbit is a closed Reeb orbit ($\R\neq0$ everywhere on $M$). Hence, locally around the vertices, $M$ looks like $\mathbb{C}^2\times S^1$. 
            \begin{figure}
                \centering
                \begin{tikzpicture}[scale=1.1]
                    \draw[thick] (0,0)--(-1.6,2.4);
                    \draw[thick,dotted] (-1.6,2.4)--(-1.8,2.7);
                    \draw[thick] (0,0)--(1.6,3.2);
                    \draw[thick,dotted] (1.6,3.2)--(1.8,3.6);
                    \draw[thick] (0,0)--(.75,3.75);
                    \draw[thick,dotted] (.75,3.75)--(.82,4.1);
                    \draw[thick] (0,0)--(-.9,3.6);
                    \draw[thick,dotted] (-.9,3.6)--(-1,4);
                    \filldraw[thick,blue,fill=blue!20,fill opacity=.5] (-1,1.5)--(1,2)--(.55,2.75)--(-.6,2.4)--cycle;
                    \node at (1,1.2) {$\mathcal{C}$};
                    \node at (-.05,2.2) {$\mu(M)$};
                    \begin{scope}[shift={(4,0)}]
                        \filldraw[thick,blue,fill=blue!20,fill opacity=.5] (0,0)--(3,0)--(2.5,1.5)--(1,1.3)--cycle;
                        \draw[thick,-stealth] (.5,.65)--(.5+.317,.65-.244);
                        \draw[thick,-stealth] (1.5,0)--(1.5,.4);
                        \draw[thick,-stealth] (2.75,.75)--(2.75-.379,.75-.126);
                        \draw[thick,-stealth] (1.75,1.4)--(1.75+.053,1.4-.396);
                        \draw[fill] (1.3,.8) circle [radius=.04];
                        \draw[dashed] (1.3,.8)--(1.3,2) node [above] {$T^3$};
                        \draw[fill] (2.125,1.45) circle [radius=.04];
                        \draw[dashed] (2.125,1.45)--(2.125,2.65) node [above] {$T^2$};
                        \draw[fill] (0,0) circle [radius=.04];
                        \draw[dashed] (0,0)--(0,1.2) node [above] {$S^1$};
                    \end{scope}
                \end{tikzpicture}
                \caption{On the left, the image $\mu(M)$ of $M$ is depicted for $n=3$ and $m=4$ as the section of $\mathcal{C}$ with the characteristic plane. On the right, we show $\mu(M)$ with the inward-pointing normals and the $T^3$-fibration degenerating to $T^2$ at the edges and $S^1$ at the vertices.}
                \label{fig--2.1.cone}
            \end{figure}
            
            The cone structure encodes some topology of the toric Sasakian manifold $M$, namely, we have $\pi_1(M)\simeq\mathbb{Z}^{n}/\sspan_{\mathbb{Z}}\{v_i\}$ and $\pi_2(M)\simeq\mathbb{Z}^{m-n}$. Hence, if $\{v_i\}$ with $\mathbb{Z}$-coefficients span all of $\mathbb{Z}^{n}$ then $M$ is simply-connected. Moreover, if $n=3$ and $\pi_1(M)=0$, then $M$ is diffeomorphic to $(S^2\times S^3)^{\# k}$, where $k=m-n$ (see \cite{Sparks:2010sn}, prop. 5.3 and cor. 5.4).
            
            Finally, note that the metric cone $C(M)$ with its symplectic structure and Kähler metric can be completely recovered from the moment map cone $\mathcal{C}$ via Kähler reduction \cite{Lerman:2001zua}: if $\{v_i\}_{i=1,\dots,m}\subset\mathbb{Z}^n$ denotes the set of normals for $\mathcal{C}$ and $\Lambda\subset\mathbb{Z}^n$ the lattice they span, then there is a map of tori
            \begin{equation}\label{eq--2.1.tori}
                T^m\simeq\mathbb{R}^m/2\pi\mathbb{Z}^m\longrightarrow\mathbb{R}^n/2\pi\Lambda
            \end{equation}
            induced by the linear map $\mathbb{R}^m\rightarrow\mathbb{R}^n,\,E_i\mapsto v_i$ taking the $i$th canonical basis vector $E_i\in\mathbb{R}^m$ to the normal $v_i$. The kernel $\mathcal{K}$ of \eqref{eq--2.1.tori} is isomorphic to $T^{m-n}\times\Gamma$ ($\Gamma$ some finite Abelian group) and we get
            \begin{equation}\label{eq--2.1.reduction}
                C(M)=\mathbb{C}^m/\!/\mathcal{K}.
            \end{equation}
            The group $T^m/\mathcal{K}\simeq T^n$ acts symplectically on $C(M)$ and the image of the moment map is precisely $\mathcal{C}$ (see \cite{Lerman:2001zua}, theorem 2.18). This way, we can express the coefficients $\{\R_i\}$ of the Reeb vector field $\R=\sum_{a=1}^{n}\R_ae_a$, $\{e_a\}\subset\mathcal{C}^\vee$ in terms of the charges $\{\omega^{(j)}_i\}^{j=1,\dots,m-n}_{i=1,\dots,m}$ with which the $j$th $S^1\subset\mathcal{K}$ acts on $\mathbb{C}^m$. Moreover, note that, for a fixed toric Sasakian manifold $M$, the smooth Sasakian metrics are parametrised by the Reeb $\R\in\mathring{\mathcal{C}}^\vee$ \cite{Martelli:2005tp}. Hence, given a cone $\mathcal{C}^\vee$, deforming $\R$ in a way that leaves it inside the cone gives again a Sasakian structure on $M$, with a squashed metric. This deformation of the Reeb will be necessary for regular Sasakian manifolds (cf. $S^5$ in \autoref{sec--5.1}) in order to obtain an equivariant theory after dimensional reduction to the 4d topologically twisted theory. The deformation enables us to tune the equivariance parameters in the 4d theory.
            
        \paragraph{Free $S^1$-Action.}
            Henceforward, we focus on the \textit{five-dimensional} case, i.e. $n=3$ and the cones $\mathcal{C}, \mathcal{C}^\vee$ are three-dimensional.    
            It was mentioned above already that $\R$ is, in general, not generating a free action on $M$. Therefore, in order to be able to quotient or reduce along a free $S^1$, we have to find some element $\vec{\X}=\sum_{i=1}^{3}\X_ie_i$ in $\mathfrak{t}_n$ whose orbits form a regular foliation $\mathcal{F}_{\vec{\X}}$ of $M$. It is obtained as the solution to the following set of equations (where $v_{m+1}=v_1$):
            \begin{equation}\label{eq--2.1.free}
                \forall i=1,\dots,m: \vec{\X}\cdot(\vec{v}_i\times\vec{v}_{i+1})=\pm1,
            \end{equation}
            and the coefficients of $\vec{\X}$ have to be integral. This ensures that the $\vec{\X}$-orbits close and pass through each foliated chart of $\mathcal{F}_{\vec{\X}}$ only once (i.e. the stabiliser is trivial everywhere on $M$). In general, \eqref{eq--2.1.free} does not have a solution for arbitrary allotments of the sign, and in many cases there is no solution such that all signs are $+$ (or $-$). However, there always is such a solution in cases where $M$ is regular: $\X=\R$ the Reeb.
            Since we are ultimately interested in 4d theories obtained as the quotient of $M$ by a free $S^1$-action, we limit our considerations to manifolds and sign allotments for which \eqref{eq--2.1.free} has solutions. 
            
            Note that, since $\vec{\R}$ lies within the dual cone, we can write $\vec{\R}=\sum_{i=1}^m\lambda_i\vec{v}_i$ with $\lambda_i>0$ (cf. \eqref{eq--2.1.dualcone}). Hence, $\vec{\R}\cdot(\vec{v}_i\times\vec{v}_{i+1})>0$ for all $i$ and therefore \eqref{eq--2.1.free} determines whether $\vec{\X}$ and $\vec{\R}$ are parallel or anti-parallel at the vertices of the polytope $\mu(M)$.

    \subsection{An Example: $Y^{p,q}$}\label{subsec.Ypq}
    
        Let us now consider the concrete example $Y^{p,q}$. These are an infinite family of closed, five-dimensional toric Sasakian\footnote{In fact, the $Y^{p,q}$ are toric Sasaki-Einstein manifolds. These can be viewed as odd-dimensional cousins of Calabi-Yau manifolds, i.e. $C(Y^{p,q})$ is CY. This additional property places some restrictions on the cone $\mathcal{C}$ which we will mention \textit{en passant}.} manifolds with $p,q\in\mathbb{Z}$ such that $p>1$, $p>q$ for which explicit metrics are known \cite{Gauntlett:2004yd,Martelli:2004wu}. For $p$, $q$ coprime, i.e. $\gcd(p,q)=1$, $Y^{p,q}$ is simply-connected and topologically $S^2\times S^3$. Geometrically, it is a principal $S^1$-bundle $Y^{p,q}\rightarrow B$ over the product of an axially squashed and a round sphere $B=S^2\times S^2$, such that the first Chern number on the standard two-cycles are $p$ and $q$. 
        
        Locally, we can choose coordinates $\theta\in[0,\pi],\phi\in[0,2\pi)$ and $y\in[y_1,y_2],\psi\in[0,2\pi)$ as polar, azimuthal angles on the two spheres and $\gamma\in[0,2\pi)$ parametrising the $S^1$-fibre. In these coordinates, a metric can be given explicitly and the Reeb vector field is
        \begin{equation}\label{eq--2.2.reeb}
            \R=3\partial_\psi-\tfrac{1}{2\ell}\partial_\gamma
        \end{equation}
        with $\ell=q(3q^2-2p^2+p(4p^2-3q^2)^{1/2})^{-1}$.
        $Y^{p,q}$ has effectively acting isometry group $SO(3)\times U(1)^2$ for $p,q$ odd and $U(2)\times U(1)$ else \cite{Martelli:2004wu}. In both cases, we have an effective Hamiltonian $T^3$-action on $Y^{p,q}$ for the following choice of basis:
        \begin{equation}\label{eq--2.2.basis}
            \begin{aligned}
                e_1=&-\partial_\phi-\partial_\psi,\\
                e_2=&\,\partial_\phi-\tfrac{l}{2}\partial_\gamma,\\
                e_3=&\,\partial_\gamma.
            \end{aligned}
        \end{equation}
        Note that the choice \eqref{eq--2.2.basis} is different from the standard Killing vectors $\partial_\phi,\partial_\psi,\partial_\gamma$ which, in general, do not generate an \textit{effective} action with closed orbits. This is due to the submanifolds $y=y_1,y_2$ being lens spaces $S^3/\mathbb{Z}_k$, $S^3/\mathbb{Z}_l$ ($k=p+q$, $l=p-q$) on which the standard basis $\partial_\phi,\partial_\gamma$ does not generate an effective action with closed orbits\footnote{See \cite{Martelli:2004wu}, section 4 for a discussion on the subtleties related to effective torus actions on lens spaces.}.
        
        In the basis above, we obtain the following moment map for $C(Y^{p,q})$ \cite{Martelli:2004wu}:
        \begin{equation}\label{eq--2.2.cone}
            \mu=(-\tfrac{r^2}{6}(1-y)(\cos\theta-1),\tfrac{r^2}{6}(1-y)\cos\theta-\tfrac{r^2}{2}l\ell y,\ell r^2y);
        \end{equation}
        here, $r$ is the cone direction.
        At the edges of the cone, $T^3$ collapses into an $S^1$; on $Y^{p,q}$ this is the case at the poles of the two spheres ($\theta=0,\pi$ and $y=y_1,y_2$), i.e. the image of \eqref{eq--2.2.cone} is spanned by four edges. By fixing, for example, $r=1$ in the cone direction and evaluating $\mu$ at the poles above, we obtain the edge vectors
        \begin{equation}\label{eq--2.2.edge}
            \vec{u}_1=[0,0,1],\quad\vec{u}_2=[0,p,-1],\quad \vec{u}_3=[p+q,-q,-1],\quad \vec{u}_4=[p-q,q-p,1]
        \end{equation}
        and \textit{inward-pointing} primitive normals\footnote{We point out that there is a vector $\vec{\xi}=[1,0,0]$ such that $\vec{\xi}\cdot\vec{v}_i=1$ for all $i$. Cones satisfying this condition for some $\vec{\xi}$ are called Gorenstein, which is the case for all Sasaki-Einstein manifolds (see e.g. \cite{Festuccia:2016gul} for a discussion).}
        \begin{equation}\label{eq--2.2.normal}
            \vec{v}_1=[1,0,0],\quad \vec{v}_2=[1,1,p],\quad \vec{v}_3=[1,2,p-q],\quad \vec{v}_4=[1,1,0].
        \end{equation}
        The cone spanned by \eqref{eq--2.2.edge} is depicted in \autoref{fig.Y21cones} for $p=2,q=1$. 
        
        The Kähler cone of $Y^{p,q}$ can also be obtained by symplectic reduction of $\mathbb{C}^4$ under a $U(1)$ \cite{Boyer:2011ia,Qiu:2013pta}. If we denote by $\{E_j\}_{j=1,\dots,4}$ the generators of the $U(1)^4$ acting canonically on $\mathbb{C}^4$, then the reduction is performed along the $U(1)$ with charge $[-p,p-q,-p,p+q]$. The Reeb vector field $\R=\sum_{i=1}^3\R_ie_i$ can also be written as $\R=\sum_{j=1}^4\omega_jE_j$ and the $\R_i$ related to the general equivariance parameters $\omega_j$:
        \begin{equation}\label{eq--2.2.equivariance}
            \R_1=-\omega_1-\omega_2-\omega_3-\omega_4,\qquad\R_2=-\omega_2-2\omega_3-\omega_4,\qquad\R_3=-p\omega_2+(q-p)\omega_3.
        \end{equation}
        For the Reeb vector field \eqref{eq--2.2.reeb} we find
        \begin{equation}
            \omega_1=\omega_3=\left(\frac{3}{2}+\frac{1}{2\ell(p-q)}\right),\qquad\omega_2=0,\qquad\omega_4=-\frac{1}{\ell(p-q)}.
            \end{equation}
        
        The Reeb vector field here is quasi-regular if $\ell$ is rational and irregular else \cite{Martelli:2004wu}. In contrast, solving \eqref{eq--2.1.free} we find that free actions are generated by
        \begin{equation}\label{eq--2.2.fibre}
            \vec{\X}^\text{ex}=[0,0,1].
        \end{equation}
        Here, $\vec{\X}^\text{ex}$ is parallel to $\R$ only at the first two vertices and anti-parallel at the last two (``ex'' for exotic; see discussion at the end of \autoref{sec--3.3} or \cite{Festuccia:2019akm}).

\section{Super Yang-Mills on Sasakian 5-Manifolds}\label{sec--3}

    In this section, we discuss $\mathcal{N}=1$ super Yang-Mills (SYM) theory on the five-dimensional toric Sasakian manifold $M$. Furthermore, we always assume that $M$ has a free direction (i.e. \eqref{eq--2.1.free} has solutions). We start by introducing the $\mathcal{N}=1$ field content and supersymmetric actions for compact simply-connected $M$ in a suitably twisted way. Then we discuss the theory on finite quotients $M/\mathbb{Z}_h$ along a free direction and finally take $h\to\infty$ to reduce to four-dimensional $\mathcal{N}=2$ SYM.

    Let us stress right from the beginning that all results in this work are computed for the cohomologically twisted theories both in 5d and 4d. If we start on a spin manifold with the physical theory, spinors solving the rigid SUGRA background \cite{Festuccia:2011ws} can be used to twist the theory. This has been done in 5d $\mathcal{N}=1$ for Sasaki-Einstein manifolds \cite{Kallen:2012va} (for which existence of Killing spinors is guaranteed, see e.g. \cite{Sparks:2010sn}) and in 4d $\mathcal{N}=2$ for any compact spin manifold \cite{Festuccia:2018rew}. The twisting can be viewed as a ``change of variables'' and the partition functions of physical and twisted theory are expected to agree. However, not all manifolds we consider in this work are spin so there might be global obstructions to writing the physical theory. On the other hand, the twisted field content are just differential forms and therefore the twisted theory still makes sense globally on all manifolds considered.

    \subsection{Simply-Connected Case}\label{sec--3.1}

        \paragraph{Cohomological Complex.}
            The 5d $\mathcal{N}=1$ vector multiplet contains as bosonic fields a gauge field $A$, a real scalar $\sigma$, and a two-form $H$; the gauginos are given in terms of a one-form $\Psi$ and a two-form $\chi$; all fields (except $A$) are valued in the adjoint of the gauge group\footnote{For simplicity, we assume $G$ to be simply-connected.} $G$. 
            The two-forms $\chi,H$ satisfy the following projection condition:
            \begin{equation}
                P^{+}\chi=\chi,\qquad P^{+} H=H.
            \end{equation}
            The projector $P_\pm=\frac{1}{2}(1\pm \iota_{\SR}\star)$ maps two-forms in $\Omega^2(M)$ to their horizontal\footnote{We can decompose a form $\omega\in\Omega^\bullet(M)$ as $\omega={\R}\wedge\iota_{\SR}\omega+\tilde{\omega}$, with $\iota_{\SR}\tilde{\omega}=0$. The component $\tilde{\omega}$ is called \textit{horizontal} and the corresponding subspace is denoted $\Omega^\bullet_H(M)$.}, (anti-)self-dual component in $\Omega^{2\pm}_H(M)$. This is the five-dimensional analogue of the 4d projector onto (anti-)self-dual two-forms.
            The cohomological complex reads
            \begin{equation}\label{eq--3.1.complex}
                \begin{array}{lll@{\hskip 4em}lll}
                   QA & = & i\Psi, & Q\Psi & = & -\iota_{\SR} F+ \DD\sigma,\\
                   Q\chi & = & H, & QH & = & -iL_{\SR}^A\chi-[\sigma,\chi],\\
                   Q\sigma & = & -i\iota_{\SR}\Psi & & &\\
                \end{array}
            \end{equation}
            and the supercharge squares to
            \begin{equation}\label{eq--3.1.Qsquare}
                Q^2=-iL_{\SR}^A+G_\sigma.
            \end{equation}
            Here, $\DD\sigma$ and $F$ denote the covariant derivative and field strength with respect to $A$ (locally, $\DD\sigma=\dd\sigma-\ii[A,\sigma]$ and $F=\dd A-\ii A^2$), $L_{\SR}^A=L_{\SR}-\ii[\iota_{\SR} A,\us]$ is the covariant Lie derivative along the Reeb and $G_\sigma$ is a gauge transformation with parameter $\sigma$:
            \begin{equation}
                G_\sigma A=\DD\sigma,\qquad G_\sigma\Phi=\ii[\sigma,\Phi]
            \end{equation}
            with $\Phi$ any field in the adjoint of the gauge group.
            
            Finally, we want to remark that a deformation of $\R$ in the cohomological formalism can be implemented easily; one simply replaces $\R$ in the complex \eqref{eq--3.1.complex} with its deformed version. This is another advantage compared to the ordinary formulation, where more background fields would have to be introduced in the (rigid) supergravity multiplet in order to preserve supersymmetry.

        \paragraph{Action and BPS Locus.}    
            In terms of the twisted fields introduced above, the SYM action can be written as
            \begin{equation}\label{eq--3.1.YMaction}
                S_\text{YM}=-\left(CS_{3,2}(A+\sigma\kappa)+\ii\tr \int\frac{1}{g_\text{YM}^2}(\kappa\wedge d\kappa\wedge\Psi\wedge\Psi)\right)+QW_\text{vec},
            \end{equation}
            where
            \begin{align}
                &CS_{3,2}(A)=\tr \int\frac{1}{g_\text{YM}^2}\kappa\wedge F\wedge F,\label{eq--3.1.cs}\\ 
                &W_\text{vec}=\tr\int\frac{1}{g_\text{YM}^2}\left(\Psi\wedge\star(-\iota_{\SR} F-\DD\sigma)-\frac{1}{2}\chi\wedge\star H+2\chi\wedge\star F+\sigma\kappa\wedge d\kappa\wedge\chi\right).
            \end{align}
            Here, $g^2_\text{YM}$ denotes the 5d YM coupling constant. Note that the terms in \eqref{eq--3.1.YMaction} written in parentheses are $Q$-closed but not $Q$-exact and can be viewed as a supersymmetric observable.
            
            The partition function for this theory can be computed using localisation techniques\footnote{A comprehensive review of such techniques can be found in \cite{Pestun:2016zxk}.}. The localisation term $t\cdot QV$ which we add to \eqref{eq--3.1.YMaction} reads
            \begin{equation}
                V=\tr\int\Psi\wedge\star(-\iota_{\SR} F-\DD\sigma)-\frac{1}{2}\chi\wedge\star H+2\chi\wedge\star F.
            \end{equation}
            In order to obtain positive kinetic terms for $\sigma, H$ we Wick-rotate them, after which one finds the simpler localisation locus: 
            \begin{equation}\label{eq--3.1.loceqns}
                F^+_H=0,\quad \iota_{\SR} F=0,\quad \DD\sigma=0.
            \end{equation}
            The first two equations can equivalently be written as the single equation\footnote{In analogy to 4d, solutions to this equation automatically satisfy the YM equation.} $\star F=-\kappa\wedge F$ and its solutions are called contact instantons. Since we are only concerned with the trivial instanton sector in this work, we restrict to zero contact instanton-number for which the localisation locus reads \cite{Qiu:2016dyj}
            \begin{equation}\label{eq.BPSlocus}
                A=0,\quad \sigma= a\in i\mathfrak{g},
            \end{equation}
            i.e. $A$ is the trivial connection and $\sigma$ a constant valued in the Lie algebra $\mathfrak{g}$ of $G$, up to gauge transformations.
            The classical piece in the localised partition function is obtained by evaluating \eqref{eq--3.1.YMaction} at the locus above and only receives a contribution from $CS_{3,2}(A+\sigma\kappa)$:
            \begin{equation}\label{eq.classicalSYM}
                S_\text{cl}=-\tr \int\frac{1}{g_\text{YM}^2} a^2\kappa\wedge d\kappa\wedge d\kappa=-\frac{8\pi^3}{g_\text{YM}^2}\varrho\tr( a^2),
            \end{equation}
            where $\varrho\equiv\vol_M\!/\!\vol_{S^5}$ and volume form of $M$ is $\frac{1}{8}\kappa\wedge(\dd\kappa)^2$ (note that $\sigma$ is imaginary after Wick-rotation). Optionally, we could multiply the right-hand side by a parameter controlling the size of $M$.

    \subsection{Taking Finite Quotients}\label{sec--3.2}
    
        The main player in this work are finite quotients of the simply-connected toric Sasakian manifolds $M$ discussed above. More specifically, we always assume\footnote{Note that there always is a locally free $S^1\subset T^3$; precisely the one that agrees with the remaining, non-degenerate $S^1$ at the vertices of the polygon, cf. \autoref{fig--2.1.cone}.} that there is a free $S^1\subset T^3$ and take a quotient by a finite Abelian subgroup $\mathbb{Z}_h$ ($h\in\mathbb{N}_{>1}$ fixed) along this $S^1$. So we have a principal $\mathbb{Z}_h$-bundle 
        \begin{equation}\label{eq--3.2.fibration}
            \begin{tikzcd}[column sep=small]
                \mathbb{Z}_h\ar[r,hook] & M\ar[r,"\pi"] & X
            \end{tikzcd}
        \end{equation}
        with $X=M/\mathbb{Z}_h$ the finite quotient. Note that $X$ is no longer simply-connected but $\pi_1(X)\simeq\mathbb{Z}_h$ (see \autoref{app-topology}) and we can view $M$ as the $h$-sheeted connected (universal) cover of $X$. This means, intuitively, that $X$ and $M$ are ``the same'' locally and differences are only noticeable on a global level. In particular, since the vector field $\X$ generating the free $S^1$ is Killing, the metric on $M$ induces a metric on $X$ such that the covering is Riemannian. 
        
        In order for the theory in \autoref{sec--3.1} to descend to $X$ we first restrict to bundles over $M$ that are equivariant under deck transformations\footnote{Note that $\gamma\in\Aut(\pi)$ is an isometry since the covering $\pi$ is Riemannian.} $\gamma\in\Aut(\pi)\simeq\pi_1(X)$ of the covering $\pi$; these bundles descend to $X$. We mainly take the view of equivariant bundles on $M$ instead of bundles on $X$ given that we have good control over the theory on $M$. All fields (except the connection) in the cohomological formulation are (endomorphism-valued) differential forms\footnote{Here, $\mathfrak{g}_P$ denotes the associated bundle $\mathfrak{g}_P=P\times_\text{ad}\mathfrak{g}$ of the gauge bundle $P$.} $\omega\in\Omega^r(M;\mathfrak{g}_P)$ and they descend to $X$ if they are $\Aut(\pi)$-invariant. Namely, let $\omega$ be such a form. Then $\omega$ descends to a differential form on $X$ if
        \begin{equation}\label{eq--3.2.descendance}
            \forall\gamma\in\Aut(\pi): \gamma^\ast\omega=\omega.
        \end{equation}
        If we restrict our theory on $M$ to fields satisfying \eqref{eq--3.2.descendance} then also the cohomological complex \eqref{eq--3.1.complex} is invariant under deck transformations\footnote{This follows from various properties of the pullback, $\gamma^\ast\dd\omega=\dd\gamma^\ast\omega$, $\gamma^\ast(\omega\wedge\eta)=\gamma^\ast\omega\wedge\gamma^\ast\eta$, $\gamma^\ast\iota_{\SR}\omega=\iota_{(\gamma^{-1})_\ast{\R}}\gamma^\ast\omega$ as well as the fact that $\gamma_\ast{\R}={\R}$ since ${\R}$ and $\X$ commute (remember ${\R}\in\mathfrak{t}_3$).} and descends to $X$. Similarly, one can show that $S_\text{YM}$ in \eqref{eq--3.1.YMaction} descends to the action on $X$, i.e. the complex is again a symmetry of the action. Hence, as we expect, local expressions can be ``pushed down'' and we obtain the supersymmetric theory on $X$.
        
        However, since $X$ is not simply-connected, we should also account for field configurations arising from this non-trivial topology. In particular, there are $h$ topologically inequivalent $U(1)$-bundles over $X$, labelled by their first Chern class $c_1$ valued in $H^2(X;\mathbb{Z})=H^2(X)$ (henceforth, if no coefficients are specified, we always assume $\mathbb{Z}$), each admitting a flat connection. Thus, for functions on $M$, condition \eqref{eq--3.2.descendance} is too strong and we should relax it such that functions on $M$ are pushed down to sections over aforementioned line bundles on $X$. Specifically, given a function $f:M\rightarrow\mathbb{C}$ on $M$, by virtue of the periodicity in $\X$-direction, we can expand
        \begin{equation}
            f(x,\alpha)=\sum_{t\in\mathbb{Z}}f_t(x)\e{\ii t\alpha},
        \end{equation}
        where $\alpha\sim\alpha+2\pi$ parametrises the $S^1$ and $x$ the other four dimensions. Since $\X$ is a linear combination of the three $S^1$-directions such that \eqref{eq--2.1.free} is satisfied, $L$ can be different from $2\pi$. Clearly, $\gamma^\ast f=f$ implies $t=0\mmod h$. Instead of restricting to such functions only, we impose the more general condition
        \begin{equation}\label{eq--3.2.projcond}
            t=c_1\mmod h,
        \end{equation}
        where we use $H^2(X;\mathbb{Z})\simeq\mathbb{Z}^{b_2(X)}\oplus\mathbb{Z}_h$ ($b_2$ the second Betti number of $X$) and the fact that $c_1$ takes values in the $\mathbb{Z}_h$-part (see \autoref{app-topology}). In this way, we can view $f$ as a section of a flat line bundle over $X$ that acquires a phase $\exp(2\pi\ii t/h)$ around a loop $\gamma\in\pi_1(X)$, which is simply the holonomy around $\gamma$ for the flat connection of the bundle. Hence, the theory still descends from $M$ to $X$ for an $S^1$-action with weights satisfying \eqref{eq--3.2.projcond}. In terms of the moment map cone $\mathcal{C}$ introduced in \eqref{eq--2.1.dualcone} this means restricting to slices
        \begin{equation}\label{eq--3.2.slices}
            \mathcal{C}_t=\{\vec{v}\in\mathcal{C}\;|\;\ev{\vec{v},\vec{\X}}=t\}
        \end{equation}
        of $\mathcal{C}$, as depicted in \autoref{fig--3.2.slices}. The orientation of $\mathcal{C}_t$ inside $\mathcal{C}$ is determined by the free direction $\vec{\X}$. Notice that this is a generalisation of the slice \eqref{eq--2.1.momentpolytope} along the Reeb. In particular, if there are $\vec v\in\mathcal{C}$ such that $\ev{\vec v,\vec\X}\le0$, the resulting polytope is no longer compact (see, e.g., \autoref{fig.S5cones} or \autoref{fig.Y21cones}).
        \begin{figure}
            \centering
            \begin{tikzpicture}[scale=1.1]
                    \draw[thick] (0,0)--(-1.8,2.7);
                    \draw[thick] (0,0)--(1.8,3.6);
                    \draw[thick] (0,0)--(.902,4.51);
                    \draw[thick] (0,0)--(-1,4);
                    \filldraw[thick,blue,fill=blue!20,fill opacity=.5] (-.5,.75)--(.5,1)--(.2575,1.375)--(-.3,1.2)--cycle;
                    \filldraw[thick,blue,fill=blue!20,fill opacity=.5] (-1,1.5)--(1,2)--(.55,2.75)--(-.6,2.4)--cycle;
                    \filldraw[thick,blue,fill=blue!20,fill opacity=.5] (-.25,.375)--(.25,.5)--(.12875,.6875)--(-.15,.6)--cycle;
                    \filldraw[thick,blue,fill=blue!20,fill opacity=.5] (-.75,1.125)--(.75,1.5)--(.4125,2.0625)--(-.45,1.8)--cycle;
                    \filldraw[thick,blue,fill=blue!20,fill opacity=.5] (-1.5,2.25)--(1.5,3)--(.825,4.125)--(-.9,3.6)--cycle;
                    \filldraw[thick,blue,fill=blue!20,fill opacity=.5] (-1.25,1.875)--(1.25,2.5)--(.6875,3.4375)--(-.75,3)--cycle;
                    \node at (1,1.2) {$\mathcal{C}$};
                    \node at (-.05,4.2) {$\mathcal{C}_{c_1+5h}$};
                    \begin{scope}[shift={(7,0)}]
                        \draw[thick] (0,0)--(-1.8,2.7);
                        \draw[thick] (0,0)--(1.8,3.6);
                        \draw[thick] (0,0)--(.902,4.51);
                        \draw[thick] (0,0)--(-1,4);
                        \filldraw[thick,blue,fill=blue!20,fill opacity=.5] (-.25,.375)--(.25,.5)--(.12875,.6875)--(-.15,.6)--cycle;
                        \filldraw[thick,blue,fill=blue!20,fill opacity=.5] (-.75,1.125)--(.75,1.5)--(.4125,2.0625)--(-.45,1.8)--cycle;
                        \filldraw[thick,blue,fill=blue!20,fill opacity=.5] (-1.25,1.875)--(1.25,2.5)--(.6875,3.4375)--(-.75,3)--cycle;
                        \node at (1,1.2) {$\mathcal{C}$};
                        \node at (-.05,3.6) {$\mathcal{C}_{c_1+2h^\prime}$};
                    \end{scope}
            \end{tikzpicture}
            \caption{Different slices $\mathcal{C}_{t}$ of the cone $\mathcal{C}$ for (left) $t=c_1+n\cdot h$ ($n=0,\dots,5$) and (right) $t=c_1+n\cdot h^\prime$ ($n=0,1,2$), where $h^\prime=2h$. Hence, as $h$ becomes larger, the slices space out. Note that here we have chosen $\vec\X$ such that $\ev{\vec v,\vec\X}>0$ for all $\vec v\in\mathbb{C}$, hence the slices are compact (cf. \autoref{fig.S5cones} and \ref{fig.Y21cones} with non-compact slices).}\label{fig--3.2.slices}
        \end{figure}
        
        Another consequence of $X$ not being simply-connected is the existence of non-trivial flat connections $A$ for the gauge bundle $G\hookrightarrow P\rightarrow X$ which now contribute to the locus \eqref{eq.BPSlocus}. It is well known that (up to gauge transformations) flat connections are in one-to-one correspondence with representations of $\pi_1(X)$ in $G$ (up to conjugation) via their holonomy $\text{h}_{A}:\pi_1(X)\rightarrow G$ around loops $[\gamma]\in\pi_1(X)$ (see e.g. \cite{Taubes:2011}). For example, $G=U(N)$ yields
        \begin{equation}\label{eq--3.16.flatconnection}
            \Hom(\pi_1(X),U(N))\simeq\{\diag(\e{2\pi\ii m_1/h},\dots,\e{2\pi\ii m_N/h})|m_i\in\mathbb{N}_{\le h}\}
        \end{equation}
        and conjugation acts simply by permuting the exponents. Hence, flat connections for $G=U(N)$ are labelled by an array of integers
        \begin{equation}\label{eq--3.2.flat}
            \mathfrak{m}\equiv\diag(m_1,\dots,m_N)\in\mathbb{N}^{N\times N}_{\le h},\qquad m_i\le m_{i+1}.
        \end{equation}
        Note that, since $\pi_1(X)$ is Abelian, all representations in $G$ will be contained in its maximal Abelian (i.e. Cartan) subgroup.
        We henceforth use the symbol $\mathfrak{m}$ to denote the elements in (the Cartan subalgebra of) $\mathfrak{g}$ labelling flat connections also for generic $G$. 
        
        Next, we revisit the BPS locus \eqref{eq.BPSlocus}. On $X$, in addition to the trivial connection, the contact instanton equation in the zero-instanton sector is solved by any flat connection. 
        In particular, the topologically non-trivial flat connections $A_\mathfrak{m}$ discussed above are part of the BPS locus now. Then, for fixed $\mathfrak{m}$, the solution to $\DD\sigma=0$ is again given by a constant,
        \begin{equation}\label{eq--3.2.quotlocus}
            \sigma=a\in\ii\mathfrak{g},\qquad [\mathfrak{m},a]=0,
        \end{equation}
        up to gauge transformations. Therefore, in general, $G$ is broken to its Cartan subgroup $U(1)^{\rk G}$ and the path integral over gauge bundles $P$ reduces to principal $U(1)^{\rk G}$-bundles which are classified precisely by $\mathfrak{m}$. Therefore, the free $S^1$ acts on fluctuations of the adjoint-valued scalar fields around the BPS locus with infinitesimal weight
        \begin{equation}\label{eq--3.18.projcond}
            t=\alpha(\mathfrak{m})\mmod h
        \end{equation}
        according to \eqref{eq--3.2.projcond} and the discussion thereafter.
        Note that the complex \eqref{eq--3.1.complex} does not mix topological sectors and hence, the form of the one-loop partition function obtained on $M$ is unchanged in each topological sector\footnote{Since fluctuations around the locus do not jump between sectors and the connections are flat, in particular, the index computation (cf. \cite{Qiu:2014oqa,Festuccia:2019akm}) that yields the one-loop determinant is unchanged.}, of course, up to imposing the descendance conditions on the fields and $S^1$-action discussed above. However, on $X$ we now have a sum of these one-loop pieces over all topological sectors, labelled by $\mathfrak{m}$. This will be further discussed in \autoref{subsec.generalformula}.

        Similar to the simply-connected case, the classical piece only receives a contribution from the Chern-Simons action $CS_{3,2}(A_{\mathfrak{m}}+a\kappa)$ on $X$:
        \begin{equation}\label{eq--3.20.quotclassical}
            S_\text{cl}(a,\mathfrak{m})=-\frac{8\pi^3}{g_\text{YM}^2}\frac{\varrho}{h}\tr( a^2),
        \end{equation}
        where we have used \eqref{eq--3.2.quotlocus} and $CS_{3,2}(A_\mathfrak{m})=0$ since $A_\mathfrak{m}$ is flat. The quotient introduces a factor of $\frac{1}{h}$. The independence of $S_\text{cl}$ from $\mathfrak{m}$ suggests that all topological sectors are weighted equally in the partition function.

    \subsection{Reduction to Four Dimensions}\label{sec--3.3}
        
        While the theory on finite quotients of $M$ is interesting in its own right, in this work we ultimately want to reduce to the 4d $\mathcal{N}=2$ theory in order to see how fluxes feature in the partition function. We achieve this by taking the order of the finite group $\mathbb{Z}_h$ we quotient by to be very large. This makes the free $S^1$ in $M$ shrink more and more until, in the limit where $h\to\infty$, the $S^1$ shrinks to a point and we are left with the base manifold $M/S^1$ which we denote by $B$ in the following (a more formal treatment of the limit can be found in \autoref{app-topology}). Note that, in general, $B$ is not toric Kähler if $\X$ is not aligned with the Reeb $\R$.

        \paragraph{Flux.}
            The bundles on $M$ that descend to $B$ are the ones equipped with connections whose holonomy around the free $S^1$ is trivial\footnote{See \cite{Festuccia:2016gul}, appendix A for a detailed discussion.} and fields in our theory on $M$ that descend to $B$ are the forms $\omega$ on $M$ that are invariant along $\X$, i.e. $\mathcal{L}_{\SX}\omega=0$ (the infinitesimal version of invariance under deck transformations discussed in \autoref{sec--3.2}). The reduction of the cohomological complex \eqref{eq--3.1.complex} and the action \eqref{eq--3.1.YMaction} is straightforward; it has been performed in \cite{Festuccia:2016gul,Festuccia:2019akm} and matches the four-dimensional cohomological complex and Yang-Mills action in \cite{Festuccia:2018rew}, except for the peculiarity that the four-dimensional Yang-Mills coupling is now position-dependent\footnote{Hence, the resulting 4d theory can be viewed as a slight generalisation of the one discussed in \cite{Festuccia:2018rew}.} (see \cite{Festuccia:2016gul} for more on this). 
            Moreover, we can rewrite the locus equations \eqref{eq--3.1.loceqns} in terms of fields pulled back from 4d via $\begin{tikzcd}[column sep=small]
                S^1\ar[r,hook] & X\ar[r,"p"] & B
            \end{tikzcd}$ using
            \begin{align}
                A&=p^\ast A_4+\varphi\, b\label{eq--3.3.A5d4d}\\
                F&=p^\ast F_4+\DD\varphi\wedge b+\varphi\,\dd b,\label{eq--3.3.F5d4d}
            \end{align} 
            where $b=g(\X/\|\X\|^2,\us)$ and the covariant derivative $\DD\varphi$ is with respect to $p^\ast A_4$ (cf. \cite{Festuccia:2019akm}, section 4). In particular, $\iota_{\SR} F=0$ produces two equations:
            \begin{equation}\label{eq--3.3.locus}
                \iota_v F_4-\DD(\iota_{\SR} b\,\varphi)=0,\qquad \iota_v\DD\varphi=0,
            \end{equation}  
            where $v=p_\ast{\R}$ is Killing with respect to the 4d metric and represents the remaining torus action. The two 4d scalars $\sigma,\varphi$ are simultaneously diagonalisable, breaking the gauge group to its Cartan subgroup $U(1)^{\rk G}$ as usual. The first equation in \eqref{eq--3.3.locus} then has solutions corresponding to line bundles\footnote{They are equivariant under the remaining torus action on $B$ represented by $v$. $F$ provides a symplectic form and $\iota_{\SR} b\,\varphi$ a moment map for the action by $v$.} on $B$ characterised by their first Chern class $c_1$. 
            
            Explicitly, the flat connections in the 5d locus on $X$ have non-trivial holonomy along the generator $[\gamma]\in\pi_1(X)$, being the loop along the $\X$-direction. Using \eqref{eq--3.3.A5d4d}, we have constant $\varphi=\mathfrak{m}$ on the locus. But since $A$ is flat, using \eqref{eq--3.3.F5d4d} yields
            \begin{equation}\label{eq--3.3.flux}
                p^\ast F_4+\varphi\dd b=0.
            \end{equation}
            Note that, while $\dd b$ is basic with respect to $\X$, it is not an exact form on $B$ and hence, the field strength $F_4$ satisfying \eqref{eq--3.3.flux} carries flux, determined by $\varphi$. In this way, flat connections in the locus on $X$ yield $F_4$ with flux in the locus of the 4d theory on $B$.

            More geometrically, remember that the flat connections on $X$ are characterised by $c_1$ with $\im c_1\simeq\mathbb{Z}_h\subset H^2(X)$. Note that $H^2(X)$ is generated by the (Poincaré duals of the) facets of the moment map polytope\footnote{To be precise, the generators are the torus fibration given by the moment map, restricted such that its image is an edge of the polytope.} $\mu(M)$, modulo some relations and the $\mathbb{Z}_h$-subgroup is generated by a particular linear combination. Consequently, quotienting along the free $S^1$ by taking $h\to\infty$ yields the generating set for $H^2(B)$; the generator for the $\mathbb{Z}_h$-subgroup yields a generator $[c]$ for a $\mathbb{Z}$-subgroup in $H^2(B)$.

            In terms of the projection condition \eqref{eq--3.18.projcond}, for $h\to\infty$ only a single mode 
            \begin{equation}\label{eq--3.3.proj}
                t=\alpha(\mathfrak{m})\in\mathbb{Z} 
            \end{equation}
            is allowed along the free $S^1$, i.e. a single slice $\mathcal{C}_t$ of the moment map cone $\mathcal{C}$ (cf. \autoref{fig--3.2.slices}); for larger values of $h$ the additional slices move further up and hide at infinity in the limit. Therefore, the sum over flat connections labelled by $\mathfrak{m}$ in \eqref{eq--3.2.flat} for the 5d partition function on $X$, in the limit, becomes a sum over fluxes on the two-cycle $c$,
            \begin{equation}
                \frac{1}{2\pi}\int_c F_4=\mathfrak{m}=\diag(m_1,\dots,m_{\rk G})\in\mathbb{Z}^{\rk G\times \rk G},
            \end{equation}
            for the 4d partition function on $B$.

        \paragraph{Instantons.}
            Another interesting feature of the 4d theory concerns instanton contributions. It was shown in \cite{Festuccia:2019akm} that the space of horizontal (with respect to $\R$), self-dual two-forms $\Omega^{2+}_H(M;\mathfrak{g}_P)$ is isomorphic to another three-dimensional subbundle of $\Omega^2_H(M;\mathfrak{g})$ which is transverse to $\X$ and given by the image of the following projector:
            \begin{equation}
                P=\frac{1}{1+g(\hat{\X},\R)^2}(1+g(\hat{\X},\R)\star_4-g_4(v,\us)\wedge\iota_v),
            \end{equation} 
            where $g$ and $g_4$ are the metrics on $M$ and $B$ and $\hat{\X}=\X/\|\X\|$.
            Note that, at the vertices of the polytope $\mu(M)$, we must have $g(\hat{\X},\R)=\pm1$ (cf. discussion succeeding \eqref{eq--2.1.free}) and $v=0$ (since the vertices are precisely the fixed points with respect to the remaining torus action). Then $P$ coincides with $P^\pm$; hence, on $B$ this happens precisely at the torus fixed points. Therefore, depending on our choice of free direction $\X$, the 4d theory localises to self-dual (SD) or anti-self-dual (ASD) connections at the fixed points and gives either instanton or anti-instanton contributions. This is a generalisation of Pestun's theory on $S^4$ and is known as \textit{Pestunization} \cite{Festuccia:2018rew}. For the examples in \autoref{sec--5} different choices of $\X$ are allowed and, consequently, we obtain different SD/ASD contributions at the 4d torus fixed points. The choice for which we have SD (or ASD) contributions at all fixed points is usually called the (\textit{anti})-\textit{topological} theory, while we call theories with mixed distributions \textit{exotic}. 

            It should be mentioned at this point that, in general, we cannot reach 4d theories with arbitrary SD/ASD distributions from 5d because there might not always be enough toric Sasakian $S^1$-fibrations over $B$ to start with (e.g., there is no such fibration for Pestun's theory on $S^4$). For such cases, the intrinsically 4d formalism in \cite{Festuccia:2018rew} has to be used. 

        \paragraph{4d Geometries.} 
            Given the reduction procedure described above for a large class of five-dimensional manifolds $M$, it is natural to ask about a classification of the possible 4d base manifolds $B$. Indeed, such a classification was provided in \cite{Festuccia:2016gul} for the special case of toric Sasaki-Einstein manifolds whose spin structure can be pushed down to $B$. Here, it was possible to write the normals of the moment map cone $\mathcal{C}$ and the free direction $\X$ in a standard form and, by computing the intersection forms of elements in $H_2(B)$, it was proved that $B$ is homeomorphic to connected sums $(S^2\times S^2)^{\#k}$, $k=m/2-1$ ($m$ the number of vertices of the polytope, which in this case is always even and $m\ge4$). 
            
            Since we do not have to impose a spin structure on $B$ (by virtue of the cohomological formulation) and the moment map cone is not required to be Gorenstein, we should be able to reach all quasi-toric 4d manifolds
            \begin{equation}\label{eq--3.3.quasitoric}
                (S^2\times S^2)^{\#a}\#(\mathbb{CP}^2)^{\#b}\#(\overline{\mathbb{CP}^2})^{\#c},\qquad a,b,c\in\mathbb{Z}_{\geq 0}.
            \end{equation}  
            However, we are not currently able to provide a proof of this statement.

\section{Partition Functions}\label{sec--4}

    In this section, we derive our main result: the one-loop partition function around fluxes of an $\mathcal{N}=2$ vector multiplet on a large class of four-dimensional manifolds and for all distributions of SD/ASD complexes at the fixed points that can be obtained from 5d. As a starting point, we take the partition function for the 5d $\mathcal{N}=1$ vector multiplet on the simply-connected manifold $M$. A detailed exposition of this result can be found in \cite{Qiu:2016dyj}. Our focus will be on the perturbative partition function $Z_M$ which involves a product over charges $\Vec{n}=(n_1,n_2,n_3)$ of the various fields under the $T^3$-action. These charges are determined by the integer-valued vectors inside the three-dimensional moment map cone $\mathcal{C}$ of $C(M)$ (cf. \eqref{eq--2.1.dualcone}). We can build up $\mathcal{C}\cap\mathbb{Z}^3$ as a collection of 2d slices (cf. \autoref{fig--3.2.slices}) labelled by the charge $t$ under the free $S^1$ instead and replace the product over $\vec{n}$ by $(n_1,n_2,t)$. Then, at fixed $t$, the perturbative partition function receives contributions from the two-dimensional slice $\mathcal{C}_t\subset\mathcal{C}$. The orientation of $\mathcal{C}_t$ inside $\mathcal{C}$ depends on the choice of free $S^1$.
    
    When moving to the quotient $X=M/\mathbb{Z}_h$ the rewriting of $Z_M$ in terms of a product over $(n_1,n_2,t)$ makes it easy to implement the projection condition \eqref{eq--3.18.projcond}. In particular, only slices $...,\mathcal{C}_{\alpha(\mathfrak{m})-h},\mathcal{C}_{\alpha(\mathfrak{m})},\mathcal{C}_{\alpha(\mathfrak{m})+h},...$ contribute to a single topological sector $\mathfrak{m}$. The partition function on $X$ then involves a sum over all these sectors.
    
    In the limit where $h\to\infty$, we reduce to the four-dimensional base $B=M/S^1$ and the flat connections on $X$ give rise to configurations with flux on $B$. Correspondingly, we arrive at the zero-instanton one-loop partition function on $B$ which, at a given flux sector, is a product over charges $(n_1,n_2)$ under the remaining $T^2$-action. These charges belong to the projection of $\mathcal{C}_{\alpha(\mathfrak{m})}$ to its first two components. 
    
    We point out that all expressions for $Z$ below are unfactorised and we leave for \autoref{sec--6} a discussion of their factorisation properties.

    \subsection{Simply-Connected Case}

        The cohomological complex for an $\mathcal{N}=1$ vector multiplet on a simply-connected toric Sasakian manifold $M$ with gauge group $G$ has been introduced in \eqref{eq--3.1.complex} and the action in \eqref{eq--3.1.YMaction}. The partition function is given by \cite{Qiu:2016dyj}
        \begin{equation}\begin{split}\label{eq.fullpartitionTS}
            \mathcal{Z}_M&=\int_{\mathfrak{h}}d a\;e^{-S_{\text{cl}}}\cdot Z_M\cdot Z_M^{\text{inst}}\\
            &=\int_{\mathfrak{h}}d a\;e^{-\frac{(2\pi )^3}{g_\text{YM}^2}\varrho\tr ( a^2)}\cdot\mbox{det}_\text{adj}'S_3^{\mathcal{C}}( \ii a |\R)\cdot\prod_{i=1}^m Z^{\text{Nek}}_{\mathbb{C}^2_{\epsilon^i_1,\epsilon^i_2}\times S^1}(a|\epsilon^i_1,\epsilon^i_2,\beta_i^{-1}).
        \end{split}\end{equation}
        The integral with respect to the Coulomb branch parameter $ a$ is over the Cartan subalgebra $\mathfrak{h}\subseteq\mathfrak{g}$ of the gauge group. Contact instantons contribute to the non-perturbative part of the partition function and it is conjectured in \cite{Qiu:2016dyj} that $Z^{\text{inst}}_M$ is obtained by gluing Nekrasov partition functions at each of the $m$ fixed fibres. In the expression above, $\epsilon^i_1,\epsilon^i_2,\beta_i^{-1}$ are the local equivariance parameters for the $T^3$-action on each neighbourhood $\mathbb{C}^2_{\epsilon^i_1,\epsilon^i_2}\times S^1$ around the fixed fibres; in particular, $\beta_i$ is the radius of the fixed fibre.
        
        We focus on the perturbative partition function\footnote{Most results on the perturbative partition function in the literature assume that $M$ is Einstein or, equivalently, that the cone $C(M)$ is Calabi-Yau. However, even without this condition, $Z_M$ is still given by the triple sine function \cite{Qiu:2016dyj}.}, written in terms of the triple sine function $S_3^\mathcal{C}(\ii\alpha(a)|\R)$ (see, e.g., \cite{Kurokawa2003MultipleSF,narukawa2004modular}) which explicitly reads
        \begin{equation}\label{eq.perturbativeTS}
            Z_M=\prod_{\alpha\in\Delta}\prod_{\Vec{n}\in \mathcal{C}\cap\mathbb{Z}^3}(\ev{\vec{n},\vec{\R}}+ \ii\alpha(a) )\prod_{\Vec{n}\in \mathring{\mathcal{C}}\cap\mathbb{Z}^3}(\ev{\vec{n},\vec{\R}}- \ii\alpha(a) ).
        \end{equation}
        Here, $(n_1,n_2,n_3)$ are the charges of the modes under rotations along the $T^3$-action and $\vec\R$ the Reeb, as usual. The moment map cone $\mathcal{C}$ of $C(M)$ and its interior $\mathring{\mathcal{C}}$ have been introduced in \eqref{eq--2.1.dualcone}. Finally, $\Delta$ denotes the root set of the gauge algebra $\mathfrak{g}$. Instead of using the $\det_\text{adj}^\prime$-notation from \eqref{eq.perturbativeTS} we henceforth write the one-loop part of the partition function explicitly as a product over the roots, omitting possible factors arising from fermionic zero-modes that would cancel a Vandermonde determinant in the integral over $a$. This way, shifts of $a$ via flat connections will be more apparent in the notation. Roughly speaking, the products in \eqref{eq.perturbativeTS} count holomorphic functions on $C(M)$, weighted by their eigenvalue of $\mathcal{L}_\R$ (the Lie-derivative with respect to the Reeb); see \cite{Qiu:2016dyj} for a derivation.

        Before we proceed to non-simply-connected spaces, we have to point out that the partition function \eqref{eq.perturbativeTS} is not the full story yet. Indeed, whenever $H^2(M)\neq0$ there might be an additional sum over flux configurations in the path integral. The fact that these configurations are not part of the localisation locus \eqref{eq--3.1.loceqns} points to the fact that more general loci should be allowed in this case. These are possibly complex or singular field configurations (see, e.g. \cite{Benini:2015noa,Closset:2015rna} in 2 and 3d). Such an analysis is outside the scope of this work, where we consider \eqref{eq.perturbativeTS} as our starting point. However, we expect our analysis to go through in the exact same way in the presence of these additional fluxes.

    \subsection{Finite Quotients and Reduction}\label{subsec.generalformula}
    
        As stated above, we consider toric Sasakian manifolds $M$ admitting a free $S^1$-action. In terms of a basis $\{e_i\}_{i=1,2,3}\subset\mathfrak{t}_3$ for an effective $T^3$-action, the vector field generating rotations along the free $S^1$ is
        \begin{equation}
            \Vec{\X}=l_1e_1+l_2e_2+l_3e_3,
        \end{equation}
        where the $\mathbb{Z}$-valued coefficients $l_1,l_2,l_3$ are obtained by solving \eqref{eq--2.1.free}; they depend on the geometry of $\mathcal{C}$ and will be discussed for concrete examples in \autoref{sec--5}. 
        
        \paragraph{Slicing the Cone.}
            The first step towards deriving the one-loop partition function around fluxes on $B=M/S^1$ consists in rewriting \eqref{eq.perturbativeTS} as a product over the slices $\mathcal{C}_t$ in \eqref{eq--3.2.slices}. Here, $t$ is the integer-valued charge under the free $S^1$ and can be written in terms of the charges $\{n_i\}$ corresponding to the initial choice of basis $\{e_i\}$ as follows:
            \begin{equation}\label{eq.tgeneric}
                t=l_1n_1+l_2n_2+l_3n_3.
            \end{equation}
            Thus, substituting\footnote{We assume $l_3\neq0$ here. Substituting for any of $n_1,n_2,n_3$ leads to the same result.} $n_3=l_3^{-1}(t-l_1n_1-l_2n_2)$ in \eqref{eq.perturbativeTS} and defining
            \begin{equation}\label{eq.epsilongeneric}
                \epsilon_1\equiv \R_1-\frac{l_1}{l_3}\R_3,\quad \epsilon_2\equiv \R_2-\frac{l_2}{l_3}\R_3,
            \end{equation}
            we find
            \begin{equation}\label{eq.tperturbativeTS1}
                Z_M= \prod_{\alpha\in\Delta}\prod_
                t\prod_{(n_1,n_2)\in\mathcal{B}_t}(n_1\epsilon_1+n_2\epsilon_2+t\tfrac{\R_3}{l_3}+ \ii\alpha(a) )\prod_{(n_1,n_2)\in\mathring{\mathcal{B}}_t}(n_1\epsilon_1+n_2\epsilon_2+t\tfrac{\R_3}{l_3}- \ii\alpha(a) ).
            \end{equation}
            Here, the region $\mathcal{B}_t$ is defined as
            \begin{equation}\label{eq.2dslice}
                \mathcal{B}_t=(\proj_{12}\mathcal{C}_t)\cap\mathbb{Z}^2=\proj_{12}\{\vec{v}\in\mathcal{C}\;|\;\ev{\vec{v},\vec{\X}}=t\}\cap\mathbb{Z}^2,
            \end{equation}
            where $\proj_{12}[v_1,v_2,v_3]=[v_1,v_2]$ projects onto the first two components and $\mathring{\mathcal{B}}_t$ denotes the interior of $\mathcal{B}_t$.
            In order to write \eqref{eq.tperturbativeTS1} more compactly, we define a slight generalisation of $\Upsilon$-functions (see, e.g., \cite{Pestun:2007rz,Hama:2012bg,Festuccia:2018rew}):
            \begin{equation}\label{eq.upsilon}
                \Upsilon^{\mathcal{B}_t}(z|\epsilon_1,\epsilon_2)=\prod_{(n_1,n_2)\in\mathcal{B}_t}(\epsilon_1 n_1+\epsilon_2 n_2+z)\prod_{(n_1,n_2)\in\mathring{\mathcal{B}}_t}(\epsilon_1 n_1+\epsilon_2 n_2+\Bar{z}),
            \end{equation}
            Using this definition, we can rewrite \eqref{eq.tperturbativeTS1} as
            \begin{equation}\label{eq.tperturbativeTS2}
                Z_M= \prod_{\alpha\in\Delta}\prod_t\Upsilon^{\mathcal{B}_t}( \ii\alpha(a) +\tfrac{\R_3}{l_3}t|\epsilon_1,\epsilon_2).
            \end{equation}

            As discussed above, we can also deform the Reeb by introducing a squashing of $M$. Let us recall that $M$ is obtained via symplectic reduction \eqref{eq--2.1.reduction} of $\mathbb{C}^m$, and that the $\R_i$  ($i=1,2,3$) can be expressed in terms of equivariance parameters $\omega_j$, ($j=1,...,m)$. These, in turn, can be squashed without spoiling the supersymmetric theory (cf. discussion in \autoref{sec--3.1}, Cohomological Complex). Hence, we impose that the squashing along the fibre vanishes, so that $\vec\X\cdot\vec\R$ stays invariant under deformations of the $\R_i$. We will use this constraint in the examples below to express $\R_3$ in terms of the equivariance parameters $\epsilon_1,\epsilon_2$.
        
        \paragraph{Taking the Quotient.}
            The next step in the construction is to pass to the quotient $X=M/\mathbb{Z}_h$ introduced in \autoref{sec--3.2}. The action of the free $S^1$ on $M$ descends to $X$ if the projection condition
            \begin{equation}
                t=\alpha(\mathfrak{m})\mmod h
            \end{equation}
            is satisfied, which is the restriction we need to impose on the modes considered in the superdeterminant. Thanks to our rewriting of $Z_M$ as \eqref{eq.tperturbativeTS2} this is achieved easily for fixed $\mathfrak{m}$:
            \begin{equation}\label{eq.perturbativeX}
                Z_X= \prod_{\alpha\in\Delta}\prod_{t=\alpha(\mathfrak{m})\mmod h}\Upsilon^{\mathcal{B}_t}( \ii\alpha(a) +\tfrac{\R_3}{l_3}t|\epsilon_1,\epsilon_2).
            \end{equation}
            This contribution is the one-loop superdeterminant around flat connections\footnote{The reason why no additional contributions appear alongside the product in \eqref{eq.perturbativeX}, according to \cite{Qiu:2016dyj}, is that $\pi_1(X)$ is pure torsion.}. The full partition function on $X$ is then a sum over the topological sectors:
            \begin{equation}
                \mathcal{Z}_X=\sum_{\mathfrak{m}}\int_{\mathfrak{h}}d a\;e^{-S_\text{cl}}\cdot Z_X\cdot Z_X^\text{inst}.
            \end{equation}
            The classical piece of the partition function is given in \eqref{eq--3.20.quotclassical}.
            We see that the slices $\mathcal{B}_t$ contributing to $Z_M$ now distribute over different topological sectors\footnote{We expect a similar treatment to be possible for the contributions from contact instantons.} according to their value of $t$. 
        
        \paragraph{Reduction to Base.}
            In order to compute the one-loop partition function around fluxes on the four-dimensional base manifold $B$, we take the limit of $h\to\infty$ as discussed in \autoref{sec--3.3}. For fixed $\mathfrak{m}$, only a single mode \eqref{eq--3.3.proj} survives and the one-loop contributions become
            \begin{equation}\label{eq.generalformula}
                Z_B= \prod_{\alpha\in\Delta}\Upsilon^{\mathcal{B}_\mathfrak{m}}( \ii\alpha(a) +\tfrac{\R_3}{l_3}\alpha(\mathfrak{m})|\epsilon_1,\epsilon_2),
            \end{equation}
            where, using the fact that $t$ only takes the single value $t=\alpha(\mathfrak{m})$, we write $\mathcal{B}_\mathfrak{m}\equiv\mathcal{B}_{t}$ here. Clearly, the infinite products in the $\Upsilon$-functions need to be regulated appropriately, for instance, using zeta function regularisation.
            The full partition function on $B$ is a sum over flux sectors:
            \begin{equation}\label{eq.fullpartitionB}
                \mathcal{Z}_B=\sum_{\mathfrak{m}}\int_{\mathfrak{h}}d a\;e^{-S_\text{cl}}\cdot Z_B\cdot Z_B^\text{inst}.
            \end{equation}
            The classical and instanton parts are determined in \autoref{sec--6}.

            Let us stress again that the sum over fluxes in \eqref{eq.fullpartitionB} might be incomplete, since our starting point was \eqref{eq.perturbativeTS}. If $H^2(M)\neq0$, there might be additional flux configurations that should already be included at the level of $\mathcal{Z}_M$. 

            We also point out that for the topological theory we sum over positive flux only (due to $t\in\mathbb{Z}_{\ge0}$) and have a product over the full root set $\Delta$ in $Z_B$. In the literature, one might instead find a sum over integer flux and a product over positive roots only. These two presentations depend on whether $Z_M$ is written as a product over all roots or rewritten in terms of positive ones; the resulting reduction $Z_B$ should be independent of this choice.

\section{Examples}\label{sec--5}
    
    In this section, we present concrete examples of the reduction procedure explained above. First, we consider the simplest toric Sasakian manifold, $S^5$, which is also Einstein and regular. Here, one of the two free directions we find is the Reeb itself, which yields the topologically twisted theory on $\mathbb{CP}^2$. The other direction results in an exotic theory. The second example, $Y^{p,q}$, is Einstein but not regular and reduces to $B$ homeomorphic to $S^2\times S^2$. We find again two free directions (one ``top'' and one ``ex''), but none of the two coincides with the Reeb. Finally, we present the manifolds $A^{p,q}$, which are neither Einstein nor regular and reduce to $(\overline{\mathbb{CP}^2})^{\#2}$. Here, we only find an ``exotic'' direction.

    \subsection{$S^5$}\label{sec--5.1}

        This example has been discussed in detail in \cite{Lundin:2021zeb} and we include it here for completeness. The metric cone is $C(S^5)=\mathbb{C}^3$. Moreover, as stated above, $S^5$ is a regular Sasaki-Einstein manifold and thus, the Reeb vector generates a free $S^1$-action which can be used to dimensionally reduce to $S^5/S^1=\mathbb{CP}^2$. In general, reducing along the Reeb results in a topologically twisted theory. This can be understood from the discussion following \eqref{eq--2.1.free}. Additionally, we find a second fibre which yields an exotic theory.
        
        Choosing the standard basis for an effectively acting $T^3\subset\mathbb{C}^3$, the edge vectors of the three-faceted cone $\mathcal{C}$ are:
        \begin{equation}
            \vec{u}_1=[0,1,0],\quad \vec{u}_2=[0,0,1],\quad \vec{u}_3=[1,0,0]
        \end{equation}    
        and the inward-pointing normals are given by:
        \begin{equation}
            \vec{v}_1=[1,0,0],\quad \vec{v}_2=[0,1,0],\quad \vec{v}_3=[0,0,1].
        \end{equation}
        Hence, $\mathcal{C}$ is simply the first octant, see \autoref{fig.S5cones}.
        The Reeb vector field can be written as follows:
        \begin{equation}
            \R=\omega_1\Vec{v}_1+\omega_2\Vec{v}_2+\omega_3\Vec{v}_3,
        \end{equation}
        where, for $S^5$, $\R_i=\omega_i\equiv 1+a_i\in\mathbb{R}$ are the equivariance parameters deformed by some parameters $a_i$. The integer-valued vectors $\Vec{n}$ in the cone $\mathcal{C}$ (see \eqref{eq--2.1.dualcone}) are found solving $\ev{\Vec{n},\Vec{v}_i}\geq 0$. They are given by:
        \begin{equation}\label{eq.dualconeS5}
            \Vec{n}\in\mathcal{C}\cap\mathbb{Z}^3=\mathbb{Z}_{\geq 0}^3.
        \end{equation}       
        Solving \eqref{eq--2.1.free}, we find two inequivalent free $S^1$-directions:
        \begin{equation}\begin{split}\label{eq.S5fibres}
            \mbox{top:}\;&\Vec{\X}^\text{top}=\Vec{v}_1+\Vec{v}_2+\Vec{v}_3=[1,1,1]\sim\R,\\
            \mbox{ex:}\;&\vec{\X}^\text{ex}=\Vec{v}_1+\Vec{v}_2-\Vec{v}_3=[1,1,-1].
        \end{split}\end{equation}
        These choices result in +++ and \tm++ distributions of SD/ASD complexes. Correspondingly, the charges of the modes under the $S^1$-rotation along the fibre are:
        \begin{equation}\begin{split}\label{eq.tS5}
            \mbox{top:}\;&t^\text{top}=n_1+n_2+n_3,\\
            \mbox{ex:}\;&t^\text{ex}=n_1+n_2-n_3,
        \end{split}\end{equation}
        where $t^\text{top}\geq 0$ while $t^\text{ex}$ ranges from $-\infty$ to $+\infty$. Notice that, since the supercharge \eqref{eq--3.1.Qsquare} squares to a translation along $\R$, reducing along the Reeb, in the undeformed case, only gives ordinary Donaldson-Witten theory with $Q^2=0$. The deformation of the Reeb used for dimensional reduction, which makes it tilt away from the fibre, is introduced here in order to obtain an equivariant four-dimensional theory also in the case of reducing along $\vec{\X}^\text{top}$.
        
        It is crucial to set the $\omega_i$ so that the deformation acts only on the base manifold while the fibre is invariant:
        \begin{equation}\begin{split}\label{eq.S5basesquashing}
            \mbox{top:} &\quad+a_1+a_2+a_3=0,\\
            \mbox{ex:} &\quad+a_1+a_2-a_3=0,
        \end{split}\end{equation}
        Then, we redefine\footnote{Notice that the equivariance parameters $\epsilon^\text{top}_{1,2}$ vanish in the undeformed limit while this is not the case for $\epsilon^\text{ex}_{1,2}$. This again shows that reducing along the Reeb results in the non-equivariant limit of Donaldson-Witten theory.}:
        \begin{equation}\label{eq:epsilon5d}\begin{array}{rll}
            \mbox{top:} &\quad\epsilon_1^\text{top}=\omega_1-\omega_3,\quad & \epsilon_2^\text{top}=\omega_2-\omega_3,\\
            \mbox{ex:} &\quad\epsilon_1^\text{ex}=\omega_1+\omega_3, &  \epsilon_2^\text{ex}=\omega_2+\omega_3.
        \end{array}\end{equation}
        Imposing \eqref{eq.S5basesquashing}, we find:
        \begin{equation}\begin{split}
            \mbox{top:}&\quad\omega_3^\text{top}=+1-\frac{\epsilon^\text{top}_1+\epsilon_2^\text{top}}{3},\\
            \mbox{ex:}&\quad\omega_3^\text{ex}=-\frac{1}{3}+\frac{\epsilon^\text{ex}_1+\epsilon_2^\text{ex}}{3}.
        \end{split}\end{equation}
        When confusion cannot arise, we henceforth drop the superscripts ``ex'' and ``top''. Finally, we substitute $n_3=\pm(t-n_1-n_1)$ in \eqref{eq.tperturbativeTS1} and obtain:
        \begin{equation}\label{eq.unfactorised.5d.asd}
            Z_{S^5}^\text{top}= \prod_{\alpha\in\Delta}\prod_{t\geq 0}\prod_{(n_1,n_2) \in\mathcal{B}_t}\bigg(\epsilon_1 n_1+\epsilon_2 n_2+ \ii\alpha(a) +\omega_3 t\bigg)\prod_{(n_1,n_2) \in\mathring{\mathcal{B}}_t}\bigg(\epsilon_1 n_1+\epsilon_2 n_2+ \ii\alpha(a) +\omega_3 t\bigg),
        \end{equation}
        \begin{equation}\label{eq.unfactorised.5d.flip}
            Z_{S^5}^\text{ex}= \prod_{\alpha\in\Delta}\prod_{t\in\mathbb{Z}}\prod_{(n_1,n_2) \in\tilde{\mathcal{B}}_t}\bigg(\epsilon_1 n_1+\epsilon_2 n_2+ \ii\alpha(a) -\omega_3 t\bigg)\prod_{(n_1,n_2) \in\tilde{\mathcal{B}}^\circ_t}\bigg(\epsilon_1 n_1+\epsilon_2 n_2+ \ii\alpha(a) -\omega_3 t\bigg).
        \end{equation}
        The slices \eqref{eq.2dslice} of the dual cone $\mathcal{C}$ are given by:
        \begin{equation}\begin{split}
            \mbox{top:} &\quad\mathcal{B}_t=\proj_{12}\{\vec{v}\in\mathcal{C}\;|\;\ev{\vec{v},\vec{\X}^\text{top}}=t^\text{top}\}\cap\mathbb{Z}^2,\\
            \mbox{ex:} &\quad
            \tilde{\mathcal{B}}_t=\proj_{12}\{\vec{v}\in\mathcal{C}\;|\;\ev{\vec{v},\vec{\X}^\text{ex}}=t^\text{ex}\}\cap\mathbb{Z}^2.
        \end{split}\end{equation}
        Using the definition of $\Upsilon^{\mathcal{B}_t}$-functions in \eqref{eq.upsilon} we can express \eqref{eq.unfactorised.5d.asd}-\eqref{eq.unfactorised.5d.flip} as:
        \begin{equation}\label{eq.ZtopS5}
            Z_{S^5}^\text{top}= \prod_{\alpha\in\Delta}\prod_{t\geq 0}\Upsilon^{\mathcal{B}_t}( \ii\alpha(a) +(1-\tfrac{\epsilon_1+\epsilon_2}{3})t|\epsilon_1,\epsilon_2),
        \end{equation}
        \begin{equation}\label{eq.ZexS5}
            Z_{S^5}^\text{ex}= \prod_{\alpha\in\Delta}\prod_{t\in\mathbb{Z}}\Upsilon^{\tilde{\mathcal{B}}_t}( \ii\alpha(a) +(\tfrac{1}{3}-\tfrac{\epsilon_1+\epsilon_2}{3})t|\epsilon_1,\epsilon_2).
        \end{equation}
        The two possible slicings $\mathcal{B}_t,\Tilde{\mathcal{B}}_t$ of $\mathcal{C}$ are depicted in \autoref{fig.S5cones}. 
        \begin{figure}[h!]
            \centering
            \tdplotsetmaincoords{80}{45}
            \begin{tikzpicture}[scale=0.75,tdplot_main_coords]
            \draw[thick,-stealth] (0,0,0) -- (5,0,0) node[anchor=north]{$\Vec{u}_3$};
            \draw[thick,-stealth] (0,0,0) -- (0,5,0) node[anchor=south]{$\Vec{u}_1$};
            \draw[thick,-stealth] (0,0,0) -- (0,0,5) node[anchor=south]{$\Vec{u}_2$};
            \draw[thick,-stealth] (0,0,0) -- (2.5,2.5,2.5) node[anchor=south]{$\Vec{\X}^\text{top}$};
            \filldraw[draw=gray,fill=gray!20,opacity=0.3]     
            (0,0,0)
            -- (5,0,0)
            -- (0,5,0)
            -- cycle;
            \filldraw[draw=gray,fill=gray!20,opacity=0.3]     
            (0,0,0)
            -- (0,5,0)
            -- (0,0,5)
            -- cycle;
            \filldraw[draw=gray,fill=gray!20,opacity=0.3]     
            (0,0,0)
            -- (5,0,0)
            -- (0,0,5)
            -- cycle;
            \filldraw[thick,draw=blue,fill=blue!20,opacity=0.5]       
            (2,0,0)
            -- (0,2,0)
            -- (0,0,2)
            -- cycle;
            \draw[thick,blue]    
            (2,0,0)
            -- (0,2,0)
            -- (0,0,2)
            -- cycle;
            \filldraw[thick,draw=red,fill=red!20,opacity=0.5]       
            (4,0,0)
            -- (0,4,0)
            -- (0,0,4)
            -- cycle;
            \draw[thick,red]        
            (4,0,0)
            -- (0,4,0)
            -- (0,0,4)
            -- cycle;
            \end{tikzpicture}
            \hspace{6em}
            \tdplotsetmaincoords{80}{45}
            \begin{tikzpicture}[scale=0.75,tdplot_main_coords]
            \draw[thick,-stealth] (0,0,0) -- (5,0,0) node[anchor=north]{$\Vec{u}_3$};
            \draw[thick,-stealth] (0,0,0) -- (0,5,0) node[anchor=south]{$\Vec{u}_1$};
            \draw[thick,-stealth] (0,0,0) -- (0,0,5) node[anchor=south]{$\Vec{u}_2$};
            \draw[thick,-stealth] (0,0,0) -- (1,1,-1) node[anchor=west]{$\Vec{\X}^\text{ex}$}; 
            \filldraw[draw=gray,fill=gray!20,opacity=0.3]     
            (0,0,0)
            -- (5,0,0)
            -- (0,5,0)
            -- cycle;
            \filldraw[draw=gray,fill=gray!20,opacity=0.3]     
            (0,0,0)
            -- (0,5,0)
            -- (0,0,5)
            -- cycle;
            \filldraw[draw=gray,fill=gray!20,opacity=0.3]     
            (0,0,0)
            -- (5,0,0)
            -- (0,0,5)
            -- cycle;
            \filldraw[draw=green,fill=green!20,opacity=0.5]       
            (0,0,0)
            -- (0,4,4)
            -- (4,0,4)
            -- cycle;
            \draw[thick,green] (0,4,4) -- (0,0,0) -- (4,0,4);
            \filldraw[draw=red,fill=red!20,opacity=0.5]       
            (0,0,2)
            -- (0,4,6)
            -- (4,0,6)
            -- cycle;
            \draw[thick,red] (0,4,6) -- (0,0,2) -- (4,0,6);
            \filldraw[draw=blue,fill=blue!20,opacity=0.5]       
            (2,0,0)
            -- (0,2,0)
            -- (0,4,2)
            -- (4,0,2)
            -- cycle;
            \draw[thick,blue] (4,0,2) -- (2,0,0) -- (0,2,0) -- (0,4,2);
            \end{tikzpicture}
            \caption{Cone $\mathcal{C}$ of $C(S^5)$. Left side: sliced along $\vec{\X}^\text{top}$ for $t=2$ (blue) and $t=4$ (red). At $t=0$ the slice only contains the origin. Right side: sliced along $\vec{\X}^\text{ex}$ for $t=2$ (blue), $t=0$ (green) and $t=-2$ (red). The slices are compact for $\vec{\X}^\text{top}$ and non-compact for $\vec{\X}^\text{ex}$. (In a slight abuse, we depict $\mathcal{C}$ and $\vec\X$ in the same ambient space.)}
            \label{fig.S5cones}
        \end{figure}
        
        Up to this point, \eqref{eq.ZtopS5} and \eqref{eq.ZexS5} are just rewritings of the perturbative partition function on $S^5$.
        The difference between the two cases arises when we introduce a quotient acting along either of the two fibres. The resulting manifold is a higher-dimensional generalisation of lens spaces $X=S^5/\mathbb{Z}_h$ and the charge of the modes under rotations along the fibre is constrained by the projection condition \eqref{eq--3.18.projcond}:
        \begin{equation}
            t=\alpha(\mathfrak{m})\mmod h.
        \end{equation}
        Thus, the one-loop partition function around flat connections is a sum over inequivalent topological sectors and, at each of them, only slices $\mathcal{B}_t,\Tilde{\mathcal{B}}_t$ satisfying the projection condition contribute.
                        
        When reducing to the base manifold $\mathbb{CP}^2$ by taking the large $h$ limit, we set $t=\alpha(\mathfrak{m})$ and find, for each flux sector:
        \begin{equation}\label{eq--5.1.cp2}
            Z_{\mathbb{CP}^2}^\text{top}= \prod_{\alpha\in\Delta}\Upsilon^{\mathcal{B}_\mathfrak{m}}( \ii\alpha(a) +(1-\tfrac{\epsilon_1+\epsilon_2}{3})\alpha(\mathfrak{m})|\epsilon_1,\epsilon_2).
        \end{equation}
        \begin{equation}
            Z_{\mathbb{CP}^2}^\text{ex}= \prod_{\alpha\in\Delta}\Upsilon^{\tilde{\mathcal{B}}_\mathfrak{m}}( \ii\alpha(a) +(\tfrac{1}{3}-\tfrac{\epsilon_1+\epsilon_2}{3})\alpha(\mathfrak{m})|\epsilon_1,\epsilon_2).
        \end{equation}
        The first expression corresponds to a +++ distribution of complexes at all three fixed points of $\mathbb{CP}^2$, and thus to an equivariant topological twisting, while, for the exotic theory, one fixed point flips to ASD and the distribution of complexes is \tm++ instead.
        
        At each flux sector on $\mathbb{CP}^2$, the charges of the modes contributing $\mathcal{B}_\mathfrak{m},\Tilde{\mathcal{B}}_\mathfrak{m}$ can be represented by projecting the slices in \autoref{fig.S5cones} to the $(n_1,n_2)$-plane. Explicitly, the integer valued vectors in the cone $\Vec{n}\in\mathcal{C}$ are determined solving $\ev{\Vec{v}_i,\Vec{n}}\geq 0$. Substituting $n_3=t^\text{top}-n_1-n_2$, the slice $\mathcal{B}_\mathfrak{m}$ is determined by:
        \begin{equation}
            n_1\geq 0,\qquad n_2\geq 0,\qquad n_1+n_2\leq t^\text{top}.
        \end{equation}
        Similarly, substituting $n_3=-t^\text{ex}+n_1+n_2$, the slice $\tilde{\mathcal{B}}_\mathfrak{m}$ is determined by:
        \begin{equation}
            n_1\geq 0,\qquad n_2\geq 0,\qquad n_1+n_2\geq t^\text{ex}.
        \end{equation} 
        The slices we obtain are depicted in \autoref{fig.CP2slices}. For the topologically twisted theory the slices are compact while those for the exotic theory are non-compact. This property is due to the complexes of the two theories being, respectively, elliptic and transversally elliptic. In the trivial flux sector, the results agree with those computed using \cite{Mauch:2021fgc}.
        \begin{figure}[h!]
            \centering
            \tdplotsetmaincoords{0}{0}
            \begin{tikzpicture}[scale=0.65,tdplot_main_coords]
            \filldraw[draw=blue,fill=blue!20,opacity=0.5]       
            (2,0,0)
            -- (0,2,0)
            -- (0,0,2)
            -- cycle;
            \draw[step=1.0,gray!60] (-.5,-.5) grid (5.2,5.2);
            \draw[thick,-stealth] (-0.5,0,0) -- (5.2,0,0) node[anchor=north]{$n_1$};
            \draw[thick,-stealth] (0,-0.5,0) -- (0,5.2,0) node[anchor=east]{$n_2$};
            \draw[thick,blue] (0,2,0) -- (2,0,0) -- (0,0,2) -- cycle;
            \end{tikzpicture}
            \hspace{6em}
            \tdplotsetmaincoords{0}{0}
            \begin{tikzpicture}[scale=0.65,tdplot_main_coords]
            \filldraw[draw=blue,fill=blue!20,opacity=0.5]       
            (2,0,0)
            -- (0,2,0)
            -- (0,5,2)
            -- (5,5,6)
            -- (5,0,2)
            -- cycle;
            \draw[step=1.0,gray!60] (-.5,-.5) grid (5.2,5.2);
            \draw[thick,-stealth] (-0.5,0,0) -- (5.2,0,0) node[anchor=north]{$n_1$};
            \draw[thick,-stealth] (0,-0.5,0) -- (0,5.2,0) node[anchor=east]{$n_2$};
            \draw[thick,blue] (0,5,2) -- (0,2,0) -- (2,0,0) -- (5,0,2);
            \end{tikzpicture}
            \caption{Slices for $S^5$. Left side: $\mathcal{B}_{\mathfrak{m}}$ of the topologically twisted theory for $\alpha(\mathfrak{m})=2$. For $\alpha(\mathfrak{m})=0$ only the origin contributes. Right side: $\tilde{\mathcal{B}}$ of the exotic theory for $\alpha(\mathfrak{m})=2$ (blue). For $\alpha(\mathfrak{m})\leq 0$ the entire quadrant contributes.}
            \label{fig.CP2slices}
        \end{figure}      

        Lastly, we point out that the one-loop partition function around fluxes for the topological theory on $\mathbb{CP}^2$ has already been computed in \cite{Bershtein:2015xfa}. The result was expressed as a sum over equivariant fluxes, compared to our partition function summing over ``physical'' fluxes. This makes a direct comparison rather difficult and it would be desirable to gain an understanding of the relation between the two results in the future.

    \subsection{$Y^{p,q}$}\label{subsec.example.Ypq}

        As a second example, we consider the infinite class of Sasaki-Einstein manifolds $Y^{p,q}$, which are either quasi-regular or irregular. They are homeomorphic to $S^2\times S^3$ and have been introduced in \cite{Gauntlett:2004yd,Martelli:2004wu} (see also \autoref{subsec.Ypq}). The perturbative partition function of the $\mathcal{N}=1$ vector multiplet on $Y^{p,q}$ has been computed in \cite{Qiu:2013pta,Qiu:2013aga,Schmude:2014lfa} and the full partition function \eqref{eq.fullpartitionTS}, including contact instanton contributions, is conjecturally obtained by gluing Nekrasov partition functions at each fixed fibre. Using our reduction procedure we will obtain the one-loop partition function around fluxes at the zero instanton sector on $B=Y^{p,q}/S^1$ (which is homeomophic to $S^2\times S^2$). As this is a new result, we give a more detailed presentation.
        
        A basis for an effective $T^3$-action has been introduced in \eqref{eq--2.2.basis}. The edge vectors $\Vec{u}_i$ and inward-pointing normals of the four-faceted cone have been computed in \eqref{eq--2.2.edge} and \eqref{eq--2.2.normal}. We also recall that the vector field generating a free $S^1$-action is obtained by solving \eqref{eq--2.1.free}, resulting in the choice of fibre \eqref{eq--2.2.fibre}:
        \begin{equation}\begin{split}\label{eq.freeactionYpq}
            \vec{\X}^\text{ex}=[0,0,1],
        \end{split}\end{equation}
        and a ++\tm\tm$\,$distribution of SD/ASD complexes on the base manifold. Correspondingly, the charges \eqref{eq.tgeneric} under rotations along the fibre are:
        \begin{equation}\begin{split}\label{eq.tYpq}
            t^\text{ex}=n_1l_1+n_2l_2+n_3l_3=n_3.
        \end{split}\end{equation}

        \paragraph{Effectively Acting $T^3$.}
            When introducing a quotient along the free $S^1$, and thus considering $X=Y^{p,q}/\mathbb{Z}_h=Y^{hp,hq}$, the four fixed points of the base manifold $B$ do not change. However, the quotient affects the submanifolds found at $y=y_1,y_2$:
            \begin{equation}
                y=y_1:\;S^3/\mathbb{Z}_{\lcm(h,k)},\quad y=y_2:\;S^3/\mathbb{Z}_{\lcm(h,l)}
            \end{equation}
            (remember $k=p+q,\,l=p-q$).      Hence, we modify the basis for the effectively acting $T^3$ \eqref{eq--2.2.basis} accordingly:
            \begin{equation}\begin{split}\label{eq.Ypq/Zhbasis}
                &e_1=-\partial_\phi-\partial_\psi,\\
                &\tilde{e}_2=\partial_\phi-\tfrac{\lcm(h,l)}{2}\partial_\gamma,\\
                &e_3=\partial_\gamma.
            \end{split}\end{equation}
            It is useful to study the matrix which relates the previous basis \eqref{eq--2.2.basis} to \eqref{eq.Ypq/Zhbasis}:
            \begin{equation}
                A=\left(\begin{array}{ccc}
                1 & 0 & 0\\
                0 & 1 & \frac{l-\lcm(h,l)}{2}\\
                0 & 0 & 1
                \end{array}\right),\quad A\cdot \left(\begin{array}{c}
                e_1  \\
                e_2  \\
                e_3
                \end{array}\right)=\left(\begin{array}{c}
                e_1  \\
                \tilde{e}_2  \\
                e_3
            \end{array}\right),\quad\mbox{det}A=1.
            \end{equation}
            For $l$ even, or for both $l$ and $h$ odd, $A\in SL(3,\mathbb{Z})$. In this case, we can keep the previous basis \eqref{eq--2.2.basis} as we always have the freedom to rotate by an $SL(3,\mathbb{Z})$-transformation. For $l$ odd and $h$ even, we can have $A\notin SL(3,\mathbb{Z})$ and need to use \eqref{eq.Ypq/Zhbasis} instead. However, we are mainly interested in the dimensional reduction to $B=Y^{p,q}/S^1$. Hence, in the following, to avoid the need of introducing \eqref{eq.Ypq/Zhbasis}, if $l$ is odd we simply choose $h$ to be odd.

        \paragraph{Example: $Y^{2,1}$.}
            In the following, we show explicitly how the one-loop partition function around fluxes on $B=Y^{2,1}/S^1$ arises from slicing the cone for the choice of fibre above.
            From the symplectic reduction in \autoref{subsec.Ypq}, we can express the coefficients $\R_i$ of the Reeb in terms of the general equivariance parameters:
            \begin{equation}\label{eq.Ypqomegas}
                \omega_1=\left(\frac{3}{2}+\frac{1}{2\ell}\right)+a_1,\quad\omega_2=a_2,\quad\omega_3=\left(\frac{3}{2}+\frac{1}{2\ell}\right)+a_3,\quad\omega_4=-\frac{1}{\ell}+a_4,
            \end{equation}
            where we have included deformations $a_i$ that will control the equivariance parameters of the dimensionally reduced 4d theory.
            Let us now repeat the procedure done earlier for $S^5$ by redefining
            \begin{equation}\begin{array}{rll}
                \epsilon_1^\text{ex}=\R_1, &  \epsilon_2^\text{ex}=\R_2,
            \end{array}\end{equation}
            and setting the deformation to act only on the base:
            \begin{equation}\begin{split}
                -2a_2-a_3=0.
            \end{split}\end{equation}
            This leads to:
            \begin{equation}\begin{split}
                \R^\text{ex}_3=-\frac{1}{2}\left(3+\frac{1}{\ell}\right).
            \end{split}\end{equation}
            We are now ready to substitute into the perturbative partition function \eqref{eq.tperturbativeTS1}. In terms of $\Upsilon^{\mathcal{B}_t}$-functions \eqref{eq.upsilon}, we find:
            \begin{equation}\label{eq.ZexY21}
                Z_{Y^{2,1}}^\text{ex}= \prod_{\alpha\in\Delta}\prod_{t\in\mathbb{Z}}\Upsilon^{{\mathcal{B}}_t}( \ii\alpha(a) +\R_3 t|\epsilon_1,\epsilon_2).
            \end{equation}
            Here, $\mathcal{B}_t$ is defined as in \eqref{eq.2dslice} using $\Vec{\X}^\text{ex}$. We show the slices $\mathcal{C}_t$ of the cone $\mathcal{C}$ and their projections $\mathcal{B}_t$ in \autoref{fig.Y21cones}.
            \begin{figure}[h!]
                \centering
                \tdplotsetmaincoords{100}{15}
                \begin{tikzpicture}[scale=0.65,tdplot_main_coords]
                \draw[thick,-stealth] (0,0,0) -- (0,0,3) node[anchor=south]{$\Vec{u}_1$};
                \draw[thick,-stealth] (0,0,0) -- (0,6,-3) node[anchor=east]{$\Vec{u}_2$};
                \draw[thick,-stealth] (0,0,0) -- (9,-3,-3) node[anchor=south]{$\Vec{u}_3$};
                \draw[thick,-stealth] (0,0,0) -- (3,-3,3) node[anchor=south]{$\Vec{u}_4$};
                \draw[thick,-stealth] (0,0,0) -- (0,0,1) node[anchor=east]{$\Vec{\X}^\text{ex}$};
                \filldraw[draw=gray,fill=gray!20,opacity=0.3]     
                    (0,0,0)
                    -- (9,-3,-3)
                    -- (3,-3,3)
                    -- cycle;
                \filldraw[draw=gray,fill=gray!20,opacity=0.3]
                    (0,0,0)
                    -- (9,-3,-3)
                    -- (0,6,-3)
                    -- cycle;
                 \filldraw[draw=gray,fill=gray!20,opacity=0.3]   
                    (0,0,0)
                    -- (0,6,-3)
                    -- (0,0,3)
                    -- cycle;
                 \filldraw[draw=gray,fill=gray!20,opacity=0.3]   
                    (0,0,0)
                    -- (0,0,3)
                    -- (3,-3,3)
                    -- cycle;
                \filldraw[ draw=green,fill=green!20,opacity=0.5] 
                    (2,-2,2)
                    -- (0,0,2)
                    -- (0,5,2)
                    %-- (6,5,2)
                    -- (6,-4,2)
                    -- cycle;
                \draw[thick,green] (6,-4,2) -- (2,-2,2) -- (0,0,2) -- (0,5,2);
                \filldraw[ draw=blue,fill=blue!20,opacity=0.5]       
                    (0,0,0)
                    -- (0,5/6*6,0)
                    %-- (6,5,0)
                    -- (3/2*4,-3/2*2,0)
                    -- cycle;
                \draw[thick,blue] (3/2*4,-3/2*2,0) -- (0,0,0) -- (0,5/6*6,0);
                \filldraw[ draw=red,fill=red!20,opacity=0.5]       
                    (3,-1,-1)
                    -- (0,2,-1)
                    -- (0,5,-1)
                    %-- (6,5,-1)
                    -- (6/5*5*1.15,-6/5*2*1.15,-1)
                    -- cycle;
                \draw[thick,red] (6/5*5*1.15,-6/5*2*1.15,-1) -- (3,-1,-1) -- (0,2,-1) -- (0,5,-1);
                \end{tikzpicture}
                \hspace{6em}
                \tdplotsetmaincoords{0}{0}
                \begin{tikzpicture}[scale=0.55,tdplot_main_coords]
                \filldraw[draw=green,fill=green!20,opacity=0.5] 
                    (2,-2,2)
                    -- (0,0,2)
                    -- (0,5,2)
                    -- (6,5,2)
                    -- (6,-4,2)
                    -- cycle;
                \draw[thick,green] (6,-4,2) -- (2,-2,2) -- (0,0,2) -- (0,5,2);
                \filldraw[ draw=blue,fill=blue!20,opacity=0.5]       
                    (0,0,0)
                    -- (0,5/6*6,0)
                    -- (6,5,0)
                    -- (3/2*4,-3/2*2,0)
                    -- cycle;
                \draw[thick,blue] (3/2*4,-3/2*2,0) -- (0,0,0) -- (0,5/6*6,0);
                \filldraw[ draw=red,fill=red!20,opacity=0.5]       
                    (3,-1,-1)
                    -- (0,2,-1)
                    -- (0,5,-1)
                    -- (6,5,-1)
                    -- (6/5*5,-6/5*2,-1)
                    -- cycle;
                \draw[thick,red] (6/5*5,-6/5*2,-1) -- (3,-1,-1) -- (0,2,-1) -- (0,5,-1);
                % \node at (-0.4,2) {$2$};
                % \node at (-0.6,-2) {$-2$};
                % \node at (2,0.4) {$2$};
                \draw[step=1.0,gray!60] (-.5,-4.2) grid (6.2,5.2);
                \draw[thick,-stealth] (-0.5,0) -- (6.2,0) node[anchor=south]{$n_1$};
                \draw[thick,-stealth] (0,-4.2) -- (0,5.2) node[anchor=east]{$n_2$};
                \draw[thick,green] (6,-4,2) -- (2,-2,2) -- (0,0,2) -- (0,5,2);
                \draw[thick,blue] (3/2*4,-3/2*2,0) -- (0,0,0) -- (0,5/6*6,0);
                \draw[thick,red] (6/5*5,-6/5*2,-1) -- (3,-1,-1) -- (0,2,-1) -- (0,5,-1);
                \end{tikzpicture}
                \caption{Left side: cone $\mathcal{C}$ of $C(Y^{2,1})$ sliced along $\vec{\X}^\text{ex}$ for $t=2$ (green), $t=0$ (blue) and $t=-1$ (red). Note that the slices are non-compact. Right side: ${\mathcal{B}}_{\mathfrak{m}}$ for $\alpha(\mathfrak{m})=0$ (blue), $\alpha(\mathfrak{m})=2$ (green) and $\alpha(\mathfrak{m})=-1$ (red). (In a slight abuse, we depict $\mathcal{C}$ and $\vec\X$ in the same ambient space.)}
                \label{fig.Y21cones}
            \end{figure}      
            
            Let us stress that, up until this point, the expression \eqref{eq.ZexY21} is just a rewriting of the perturbative partition function on $Y^{2,1}$. Expressions start to differ once we take quotients along $\vec\X^\text{ex}$. As explained previously, the one-loop partition function around flat connections is given by a sum over topological sectors (see \eqref{eq.perturbativeX}). Only those slices $\mathcal{B}_t$ satisfying the projection condition \eqref{eq--3.18.projcond} contribute to a given topological sector.
            
            Upon reducing to the base manifold $B\simeq S^2\times S^2$ by taking the large $h$ limit, we set $t=\alpha(\mathfrak{m})$ and obtain, for each flux sector,
            \begin{equation}\label{eq.ZexB21}
                Z_{B}^\text{ex}= \prod_{\alpha\in\Delta}\Upsilon^{{\mathcal{B}}_\mathfrak{m}}( \ii\alpha(a) +\R_3\alpha(\mathfrak{m})|\epsilon_1,\epsilon_2).
            \end{equation}
            The contributions $\mathcal{B}_\mathfrak{m}$ for each flux sector are obtained by projecting the slices in \autoref{fig.Y21cones} to the $(n_1,n_2)$-plane. Explicitly, the integer-valued vectors in the cone $\Vec{n}\in\mathcal{C}$ are determined by solving $\ev{\Vec{v}_i,\Vec{n}}\geq 0$. Substituting $n_3=t^\text{ex}$, the slice $\mathcal{B}_\mathfrak{m}$ is determined by
            \begin{equation}
                n_1\geq 0,\qquad n_1+n_2\geq -2t^\text{ex},\qquad n_1+2n_2\geq -t^\text{ex},\qquad n_1+n_2\geq 0.
            \end{equation}
            The slices are shown in \autoref{fig.Y21cones}. As for $S^5$, we obtain non-compact slices for the exotic theory. Again, this property is due to the complex of the theory being transversally elliptic. The result for the trivial flux sector agrees with the one computed in \cite{Festuccia:2016gul,Festuccia:2018rew,Mauch:2021fgc}.

        \paragraph{General Case.}
            For arbitrary $p,q$ the choice of fiber and the corresponding charges of the modes have been introduced in \eqref{eq.freeactionYpq}-\eqref{eq.tYpq}. Then, the (squashed) equivariance parameters read:
            \begin{equation}
                \omega_1=\left(\frac{3}{2}+\frac{1}{2(p-q)\ell}\right)+a_1,\quad\omega_2=a_2,\quad\omega_3=\left(\frac{3}{2}+\frac{1}{2(p-q)\ell}\right)+a_3,\quad\omega_4=\frac{1}{(q-p)\ell}+a_4.
            \end{equation}
            As in the previous examples, we define
            \begin{equation}\label{eq.epsilonsYpq}\begin{array}{rll}
                \epsilon_1^\text{ex}=\R_1, &  \epsilon_2^\text{ex}=\R_2,
            \end{array}\end{equation}
            and set the deformation to act only on the base:
            \begin{equation}\begin{split}\label{eq.Ypqbasesquashing}
                -pa_2+(q-p)a_3=0,
            \end{split}\end{equation}
            which is equivalent to
            \begin{equation}\begin{split}\label{eq.R3exYpq}
                \R_3^\text{ex}=-\frac{3}{2}(p-q)-\frac{1}{2\ell}.
            \end{split}\end{equation}
            Finally, substituting into \eqref{eq.perturbativeTS} and using the definition of $\Upsilon^{\mathcal{B}_t}$-functions in \eqref{eq.upsilon}:
            \begin{equation}\label{eq.ZexYpq}
                Z_{Y^{p,q}}^\text{ex}= \prod_{\alpha\in\Delta}\prod_{t\in\mathbb{Z}}\Upsilon^{{\mathcal{B}}_t}( \ii\alpha(a) +\R_3t|\epsilon_1,\epsilon_2),
            \end{equation}
            where, as for $Y^{2,1}$, $\mathcal{B}_t$ are defined in \eqref{eq.2dslice}.
            
            The one-loop partition function around flat connections on $X$ for the exotic theory is obtained simply by imposing the projection condition \eqref{eq--3.18.projcond} and summing over flat connections. Finally, taking the limit of large $h$, we find the one-loop partition function around fluxes on $B=Y^{p,q}/S^1$:
            \begin{equation}\label{eq.ZexBpq}
                Z_{B}^\text{ex}= \prod_{\alpha\in\Delta}\Upsilon^{{\mathcal{B}}_{\mathfrak{m}}}( \ii\alpha(a) +\R_3\alpha(\mathfrak{m})|\epsilon_1,\epsilon_2).
            \end{equation}
            Again, we obtain the slicings by studying the conditions on $(n_1,n_2)$, looking for solutions of $\ev{\Vec{v}_i,\Vec{n}}\geq 0$ at given $t$. Substituting $n_3=t^\text{ex}$, we find for ${\mathcal{B}}_\mathfrak{m}$:
            \begin{equation}
                n_1\geq 0,\qquad n_1+n_2\geq -p\,t^\text{ex},\qquad n_1+2n_2\geq -(p-q)\,t^\text{ex},\qquad n_1+n_2\geq 0.
            \end{equation}
            As a consistency check, one can see that these reduce to \eqref{eq.ZexB21} for $p=2,\;q=1$. The slices $\mathcal{B}_\mathfrak{m}$ depend explicitly on $p,q$. Although we do not have a proof, we believe that this dependence cannot be entirely removed. Notice that this is not in disagreement with \cite{Festuccia:2018rew}, where it is shown that infinitesimal deformations of the metric enter the Lagrangian through $\delta$-exact terms. But different values of $p,q$ correspond to variations of the metric on $B\simeq S^2\times S^2$ that are not connected to the identity.

    \subsection{$T^{1,1}$}

        As another example, let us consider the space $T^{1,1}$ whose metric cone $C(T^{1,1})$ is the conifold. The latter is obtained as the symplectic quotient $\mathbb{C}^4/\!/U(1)$ where the $U(1)$ acts with weights $(1,-1,1,-1)$. Details about $T^{1,1}$ can be found, e.g., in \cite{Martelli:2004wu}. Similarly to $S^5$, $T^{1,1}$ is Sasaki-Einstein, and the Reeb vector field generates a regular foliation whose leaf space is $S^2\times S^2$. 

        The inward-pointing normals of the corresponding moment map cone are given by
        \begin{equation}
            \vec{v}_1=[-1,0,1],\quad \vec{v}_2=[0,-1,1],\quad \vec{v}_3=[1,0,0],\quad \vec{v}_4=[0,1,0].
        \end{equation}
        Apart from the free direction $\vec\X^\text{top}=[0,0,1]$ proportional to the Reeb $\vec\R=[0,0,\frac{3}{2}]$, which corresponds to the ++++ distribution (there is also an exotic direction $\X^\text{ex}=[0,2,-1]$ corresponding to ++\tm\tm, which we leave as an exercise). Similar to the previous examples, we deform the Reeb vector field slightly to $\vec\R=[\omega_1,\omega_2,\omega_3]$, where $\omega_{1,2}\ll1$ and $\omega_3-\frac{3}{2}\ll1$. Keeping $\vec\R\cdot\vec\X$ constant under this deformation yields $\omega_1+\omega_2+\omega_3=\frac{3}{2}$, from which we can determine $\omega_3$. Then the one-loop partition function on the base $B=S^2\times S^2$ for fixed $\mathfrak{m}$ reads
        \begin{equation}
            Z^\text{top}_{S^2\times S^2}=\prod_{\alpha\in\Delta}\Upsilon^{\mathcal{B}_{\mathfrak{m}}}\left(\ii\alpha(a)+\left(\frac{3}{2}-\epsilon_1-\epsilon_2\right)\alpha(\mathfrak{m})\bigg|\epsilon_1,\epsilon_2\right),
        \end{equation}
        where $\epsilon_i\equiv\omega_i$ and the slices $\mathcal{B}_{t}$ are determined by the following inequalities:
        \begin{equation}
            t^\text{top}-n_1\ge0,\qquad t^\text{top}-n_2\ge0,\qquad n_{1,2}\ge0
        \end{equation}
        for all $(n_1,n_2)\in\mathcal{B}_t$ and $t^\text{top}=n_3$.

    \subsection{$A^{p,q}$}

        In this last example we consider a class of manifolds which are not Einstein (in contrast to $S^5,Y^{p,q}$). In particular, this means that the moment map cone $\mathcal{C}$ is not Gorenstein in this case. This property was exploited in \cite{Festuccia:2016gul} in view of retaining a spin structure on the base $B$ after reduction, which is needed for the ordinary formulation of the vector multiplet supersymmetry and for both formulations of the hypermultiplet one. As discussed in \autoref{sec--3.1}, we can drop this condition without any consequence.

        Apart from the Einstein property, the Sasakian manifolds $A^{p,q}$ with $p>q>0$ differ from $Y^{p,q}$ by their base space $A^{p,q}/S^1\simeq(\overline{\mathbb{CP}^2})^{\#2}$.
        The edge vectors of the moment map cone are given by
        \begin{equation}
            \vec{u}_1=[0,0,1],\quad \vec{u}_2=[q,0,1],\quad \vec{u}_3=[2p-q,p-q,1],\quad \vec{u}_4=[0,p,-1],
        \end{equation}
        along with the inward-pointing normals:
        \begin{equation}
            \vec{v}_1=[1,0,0],\quad \vec{v}_2=[0,1,0],\quad \vec{v}_3=[-1,2,q],\quad \vec{v}_4=[-1,1,p].
        \end{equation}
        Furthermore, we identify the fibre that generates a free $S^1$-action by solving \eqref{eq--2.1.free}:
        \begin{equation}\begin{split}\label{eq.Apq.ex}
            \mbox{ex:} &\quad\vec{\X}^\text{ex}=[0,0,1],
        \end{split}\end{equation}
        corresponding to a +++\tm$\,$ distribution of SD/ASD complexes on the base manifold. The charge under rotations along the fibre is given by $t^\text{ex}=n_3$. Similar to the previous example, we only find a single solution to \eqref{eq--2.1.free} and we are unable to access the topological theory on $B$ using the five-dimensional procedure. 
        
        For simplicity, in the following calculations we do not turn on a deformation of the Reeb vector and thus we simply relabel $\epsilon^\text{ex}_1=\R_1,\epsilon^\text{ex}_2=\R_2$. Substituting in \eqref{eq.perturbativeTS} gives:
        \begin{equation}\label{eq.Zex.Apq}
                Z_{A^{p,q}}^\text{ex}= \prod_{\alpha\in\Delta}\prod_{t\in\mathbb{Z}}\Upsilon^{{\mathcal{B}}_t}\big( \ii\alpha(a) +\R_3 t|\epsilon_1,\epsilon_2\big).
        \end{equation}
        The slices $\tilde{\mathcal{B}}_t$ are shown on the left in \autoref{fig.A21slices} for the case of $p=2,q=1$. As in the previous examples, the slices are normal to the vector $\vec{\X}^\text{ex}$.
        \begin{figure}[h!]
            \centering
            \tdplotsetmaincoords{100}{15}
            \begin{tikzpicture}[scale=0.75,tdplot_main_coords]
            \draw[thick,-stealth] (0,0,0) -- (0,0,2) node[anchor=south]{$\Vec{u}_1$};
            \draw[thick,-stealth] (0,0,0) -- (2,0,2) node[anchor=south]{$\Vec{u}_2$};
            \draw[thick,-stealth] (0,0,0) -- (6,2,2) node[anchor=south]{$\Vec{u}_3$};
            \draw[thick,-stealth] (0,0,0) -- (0,4,-2) node[anchor=east]{$\Vec{u}_4$};
            \draw[thick,-stealth] (0,0,0) -- (0,0,1) node[anchor=east]{$\vec{\X}^\text{ex}$};
            \filldraw[draw=gray,fill=gray!20,opacity=0.3]     
                (0,0,0)
                -- (6,2,2)
                -- (0,4,-2)
                -- cycle;
            \filldraw[draw=gray,fill=gray!20,opacity=0.3]
                (0,0,0)
                -- (6,2,2)
                -- (2,0,2)
                -- cycle;
             \filldraw[draw=gray,fill=gray!20,opacity=0.3]   
                (0,0,0)
                -- (2,0,2)
                -- (0,0,2)
                -- cycle;
             \filldraw[draw=gray,fill=gray!20,opacity=0.3]   
                (0,0,0)
                -- (0,0,2)
                -- (0,4,-2)
                -- cycle;
            \filldraw[ draw=green,fill=green!20,opacity=0.5] 
                (1,0,1)
                -- (0,0,1)
                -- (0,5,1)
                -- (9/2,5/2,1)
                -- (3,1,1)
                -- cycle;
            \draw[thick,green] (9/2,5/2,1)
                -- (3,1,1)
                -- (1,0,1)
                -- (0,0,1)
                -- (0,5,1);
            \filldraw[ draw=blue,fill=blue!20,opacity=0.5]       
                (0,0,0)
                -- (0,5/6*6,0)
                -- (3,3,0)
                -- cycle;
            \draw[thick,blue] (3,3,0) -- (0,0,0) -- (0,5/6*6,0);
            \filldraw[ draw=red,fill=red!20,opacity=0.5]       
                (0,2,-1)
                -- (0,5,-1)
                -- (3/2,7/2,-1)
                -- cycle;
            \draw[thick,red] (3/2,7/2,-1) -- (0,2,-1) -- (0,5,-1);
            \end{tikzpicture}
            \hspace{6em}
            \tdplotsetmaincoords{0}{0}
            \begin{tikzpicture}[scale=0.65,tdplot_main_coords]
            \filldraw[ draw=green,fill=green!20,opacity=0.5] 
                (1,0,1)
                -- (0,0,1)
                -- (0,5,1)
                -- (5/2+9/2,5/2+5/2,1)
                -- (3,1,1)
                -- cycle;
            \filldraw[ draw=blue,fill=blue!20,opacity=0.5]       
                (0,0,0)
                -- (0,5/6*6,0)
                -- (5/3*3,5/3*3,0)
                -- cycle;
            \filldraw[ draw=red,fill=red!20,opacity=0.5]       
                (0,2,-1)
                -- (0,5,-1)
                -- (3,5,-1)
                -- cycle;
            % \node at (-0.4,2) {$2$};
            % \node at (1,-0.4) {$1$};
            % \node at (3,-0.4) {$3$};
            \draw[step=1.0,gray!60] (-.5,-.5) grid (7.2,5.2);
            \draw[thick,-stealth] (-0.5,0) -- (7.2,0) node[anchor=south]{$n_1$};
            \draw[thick,-stealth] (0,-0.5) -- (0,5.2) node[anchor=east]{$n_2$};
            \draw[thick,green] (5/2+9/2,5/2+5/2,1)
                -- (3,1,1)
                -- (1,0,1)
                -- (0,0,1)
                -- (0,5,1);
            \draw[thick,blue] (5/3*3,5/3*3,0) -- (0,0,0) -- (0,5/6*6,0);
            \draw[thick,red] (3,5,-1) -- (0,2,-1) -- (0,5,-1);
            \end{tikzpicture}
            \caption{$\mathcal{C}$ and ${\mathcal{B}}_\mathfrak{m}$ of $A^{2,1}$. Left side: cone $\mathcal{C}$ sliced along $\vec{\X}^\text{ex}$ for $t=1$ (green), $t=0$ (blue) and $t=-1$ (red). Note that the slices are non-compact. Right side: ${\mathcal{B}}_{\mathfrak{m}}$ for $\alpha(\mathfrak{m})=1$ (green), $\alpha(\mathfrak{m})=0$ (blue) and $\alpha(\mathfrak{m})=-1$ (red). (In a slight abuse, we depict $\mathcal{C}$ and $\vec\X$ in the same ambient space.)}
            \label{fig.A21slices}   
        \end{figure}
        
        We now introduce a quotient by $\mathbb{Z}_h$ acting on the fibre. This introduces a sum over flat connections and demands that we impose the projection condition on $t^\text{ex}$. Taking the large $h$ limit we find the one-loop partition function around fluxes on the base manifold:
        \begin{equation}\label{eq.Zex.A21.B21}
                Z_{B}^\text{ex}= \prod_{\alpha\in\Delta}\Upsilon^{{\mathcal{B}}_\mathfrak{m}}\big( \ii\alpha(a) +\R_3 \alpha(\mathfrak{m})|\epsilon_1,\epsilon_2\big).
            \end{equation}
        The form of the partition function is identical to \eqref{eq.ZexB21}, however, the slices ${\mathcal{B}}_{\mathfrak{m}}$ are different. They are found solving $\vec{v}_i\cdot\vec{n}\geq 0$:
        \begin{equation}
                n_1\geq 0,\qquad n_2\geq 0,\qquad -n_1+2n_2\geq -q\,t^\text{ex},\qquad -n_1+n_2\geq -p\,t^\text{ex}
        \end{equation}
        and we display them on the right hand side of \autoref{fig.A21slices}, again for the case of $p=2,q=1$. Notice that, in this example, we have not expressed $\R_3$ in terms of $\epsilon_1,\epsilon_2$. To achieve this we would need to study how the Kähler cone of $A^{p,q}$ can be obtained by symplectic reduction of $\mathbb{C}^4$. This would allow us to relate the components of the Reeb $\R_i$ to general equivariance parameters $\omega_j$, as in \eqref{eq--2.2.equivariance} for $Y^{p,q}$. This, in turn, would enable us to introduce a deformation of $\R$ acting on the base $B$ only and, finally, to write $\R_3$ in terms of $\epsilon_1,\epsilon_2$.

\section{Factorised Partition Functions}\label{sec--6}

    Given the equivariance in our setup, it is natural to ask whether the one-loop around fluxes can be expressed as a product of local contributions on a neighbourhood $\mathbb{C}^2_{\epsilon^i_1,\epsilon^i_2}\times S^1$ around the fixed fibres in 5d and, correspondingly, on a neighbourhood $\mathbb{C}^2_{\epsilon^i_1,\epsilon^i_2}$ around the torus fixed points in 4d. This factorisation property was confirmed for the perturbative partition function in \cite{Qiu:2014oqa} and \cite{Festuccia:2019akm}, respectively. 

    While, in the 5d case, the usual way of factorising $Z_M$ is by expressing the triple-sine function in terms of $q$-Pochhammer symbols, in our case we need to factorise $\Upsilon$-functions instead, due to the slicing of the cone. Their factorisation has been discussed in \cite{Festuccia:2019akm} and requires giving an imaginary part to the hitherto real vector field $\X$ and equivariance parameters $\epsilon_1,\epsilon_2$. Then we have the following factorisation property: 
    \begin{equation}\label{eq.factorisedTS}
        \prod_t\Upsilon^{\mathcal{B}_t}( \ii\alpha(a) +\tfrac{\R_3}{l_3}t|\epsilon_1,\epsilon_2)=\prod_{i=1}^m \prod_{t\in\mathbb{Z}}\Upsilon_i( \ii\alpha(a) +\beta_i^{-1}t|\epsilon^i_1,\epsilon^i_2)^{s_i},
    \end{equation}
    where $i$ labels the $m$ fixed fibres and $\epsilon^i_1,\epsilon^i_2,\beta_i^{-1}$ are the local equivariance parameters for the $T^3$-action on $\mathbb{C}^2_{\epsilon^i_1,\epsilon^i_2}\times S^1$, given by
    \begin{equation}
    	\beta_i^{-1}=(\Vec{v}_i\times\Vec{v}_{i+1})\cdot\vec{\R},\quad\epsilon^i_1=\frac{(\vec{\R}\times\Vec{v}_{i+1})\cdot\vec{\X}}{(\vec{v}_i\times\vec{v}_{i+1})\cdot\vec{\X}},\quad\epsilon^i_2=\frac{(\Vec{v}_i\times\vec{\R})\cdot\vec{\X}}{(\vec{v}_i\times\vec{v}_{i+1})\cdot\vec{\X}}.
    \end{equation}
    The $\Upsilon_i$-functions are defined as follows:
    \begin{equation}
        \Upsilon_i(z|\epsilon_1,\epsilon_2)=\prod_{(j,k)\in\mathcal{D}_i}(\epsilon_1 j+\epsilon_2 k+z)\prod_{(j,k)\in\mathcal{D}_i^\prime}(\epsilon_1 j+\epsilon_2 k+\bar{z}).
    \end{equation}
    Note that this is essentially a generalisation of \eqref{eq.upsilon} where the second product now is over the region $\mathcal{D}^\prime_i$ which does not necessarily coincide with the interior of $\mathcal{D}_i$.
    The regions $\mathcal{D}_i,\mathcal{D}^\prime_i$ depend on the imaginary parts of $\epsilon_1^{i},\epsilon_2^{i}$ in the following way:
    \begin{equation}\begin{split}\label{eq.regularisation}
        \imp(\epsilon^{i}_1)>0,\; \imp(\epsilon_2^{i})>0:\quad&\mathcal{D}=\{(j,k)\in\mathbb{Z}^2\;|\;j\geq 0 \mbox{ and } k\geq 0\},\quad\quad\,\; s_i=+1,\\
        &\mathcal{D}^\prime=\{(j,k)\in\mathbb{Z}^2\;|\;j\geq 1 \mbox{ and } k\geq 1\},\\
        \imp(\epsilon^{i}_1)>0,\; \imp(\epsilon^{i}_2)<0:\quad&\mathcal{D}=\{(j,k)\in\mathbb{Z}^2\;|\;j\geq 0 \mbox{ and } k\leq -1\},\quad\;\;\, s_i=-1,\\
        &\mathcal{D}^\prime=\{(j,k)\in\mathbb{Z}^2\;|\;j\geq 1 \mbox{ and } k\leq 0\},\\
        \imp(\epsilon^{i}_1)<0,\; \imp(\epsilon^{i}_2)<0:\quad&\mathcal{D}=\{(j,k)\in\mathbb{Z}^2\;|\;j\leq -1 \mbox{ and } k\leq -1\},\quad s_i=+1,\\
        &\mathcal{D}^\prime=\{(j,k)\in\mathbb{Z}^2\;|\;j\leq 0 \mbox{ and } k\leq 0\}\\
        \imp(\epsilon^{i}_1)<0,\; \imp(\epsilon^{i}_2)>0:\quad&\mathcal{D}=\{(j,k)\in\mathbb{Z}^2\;|\;j\leq -1 \mbox{ and } k\geq 0\},\quad\;\;\, s_i=-1,\\
        &\mathcal{D}^\prime=\{(j,k)\in\mathbb{Z}^2\;|\;j\leq 0 \mbox{ and } k\geq 1\}.
    \end{split}\end{equation}
    The different choices in \eqref{eq.regularisation} are known as different regularisations at the fixed fibres.
    Note, in particular, that the sign of $s_i$ can be negative at the fixed fibres, for which the respective contributions of $\Upsilon_i$ appear in the denominator (hence, from the superdeterminant viewpoint, they can be interpreted as bosonic modes remaining after cancellations, instead of fermionic ones). We choose an arbitrary sign for the imaginary part of the equivariance parameters at one fixed fibre, from which the signs at all other fixed fibres follow. It can be shown \cite{Festuccia:2019akm} that the partition function is independent of this initial choice of sign.

    Once we consider the quotient $X=M/\mathbb{Z}_h$, one can still glue contributions on $\mathbb{C}^2_{\epsilon^i_1,\epsilon^i_2}\times S^1$ from the fixed fibres as in \eqref{eq.factorisedTS}, but has to impose the projection condition \eqref{eq--3.18.projcond} for the charge under rotations along the free $S^1$. Hence, for each topological sector:
    \begin{equation}
        Z_X=\prod_{\alpha\in\Delta}\prod_{i=1}^m \prod_{t=\alpha(\mathfrak{m})\mmod h}\Upsilon_i( \ii\alpha(a) +\beta_i^{-1}t|\epsilon^i_1,\epsilon^i_2)^{s_i}
    \end{equation}
    and, upon reduction, we can express the one-loop partition function around fluxes on $B$ in a factorised form as a product over the torus fixed points:  
    \begin{equation}\label{eq--6.Bpert}
        Z_B=\prod_{\alpha\in\Delta}\prod_{i=1}^m \Upsilon_i( \ii\alpha(a) +\beta_i^{-1}\alpha(\mathfrak{m})|\epsilon^i_1,\epsilon^i_2)^{s_i}.
    \end{equation}
    Now we can simply read out the shifts of the Coulomb branch parameter by the flux contributions at each fixed point and write down the classical and instanton contributions to $\mathcal{Z}_B$ explicitly. Using equivariant localisation, the classical piece gives
    \begin{equation}\label{eq--6.classical}
        e^{-S_\text{cl}}=\exp\left(-\sum_{i=1}^m\frac{(2\pi)^2}{g_\text{YM,4d}^2(x_i)}\frac{\tr a^2}{\epsilon_1^i\,\epsilon_2^i}\right),
    \end{equation}
    with $g_\text{YM,4d}(x_i)$ the (position-dependent) 4d Yang-Mills coupling evaluated at the fixed points $x_i$ (see \cite{Festuccia:2016gul} for a detailed derivation of \eqref{eq--6.classical} starting from \eqref{eq.classicalSYM}, respectively\footnote{In the limit $h\to\infty$ we keep the product $g_\mathrm{YM}^2\cdot h$ fixed, where $g_\mathrm{YM}$ is the 5d YM coupling.} \eqref{eq--3.20.quotclassical}). Note that \eqref{eq--6.classical} has no flux-dependence\footnote{However, matching for example our result for the topologically twisted theory on $\mathbb{CP}^2$ with the one involving equivariant fluxes in \cite{Bershtein:2015xfa} requires, in each flux sector, a shift of $a$ by $\mathfrak{m}$. This would indeed introduce the flux-dependence expected from a 4d perspective.} which we anticipate, since already $S_\mathrm{cl}$ on $X$ is independent of the topological class of flat connections.
    
    The instanton piece is obtained as the standard product of Nekrasov partition functions on $\mathbb{C}^2_{\epsilon^i_1,\epsilon^i_2}$ over the fixed points, applying the appropriate shifts to $a$ \cite{Nekrasov:2003vi,Festuccia:2018rew}. For $\X$ such that the cohomologically twisted background on $B$ localises to instantons at $r$ of the $m$ fixed points and anti-instantons at the remaining ones, we obtain
    \begin{equation}\label{eq--6.Binst}
        Z_B^\text{inst}=\prod_{i=1}^r Z^{Nek}_{\mathbb{C}^2}(\ii a+\beta_i^{-1}\mathfrak{m}|\epsilon^i_1,\epsilon^i_2,q)\prod_{i=r+1}^m Z^{Nek}_{\mathbb{C}^2}(\ii a+\beta_i^{-1}\mathfrak{m}|\epsilon^i_1,\epsilon^i_2,\bar{q}),
    \end{equation}
    where $q=\exp(2\pi\ii\tau)$ is the usual instanton counting parameter.

    \paragraph{Example: $Y^{p,q}$.}
        The local equivariance parameters for corresponding to the free direction $\vec\X$ are considered in \autoref{subsec.example.Ypq} and are shown in \autoref{tab--6}.
        \begin{table}[h]
            \centering
            \caption{Local equivariance parameters for the exotic theory.}\label{tab--6}
            \begin{tabular}{ c || c | c | c | c }
                %\hline
                $i$ & 1 & 2 & 3 & 4\\ \hline
                & & & & \\[\dimexpr-\normalbaselineskip+3pt]
                $\epsilon_1^{\text{ex},i}$ & $\epsilon_1-\epsilon_2$ & $2\epsilon_1-\epsilon_2$ & $\epsilon_2-\epsilon_1$ & $\epsilon_2$\\
                $\epsilon_2^{\text{ex},i}$ &  $\epsilon_2 $ & $ \epsilon_2-\epsilon_1$ & $ 2\epsilon_1-\epsilon_2 $ & $\epsilon_1-\epsilon_2$\\ 
                $\beta_i^{-1}$ & $\R^\text{ex}_3$ &  $\left(p-q\right)\left(\epsilon_1-\epsilon_2\right)+\R^\text{ex}_3 $ & $ \left(p+q\right)\epsilon_1-q\epsilon_2-\R^\text{ex}_3 $ & $ p\epsilon_2-\R^\text{ex}_3 $
            \end{tabular}
        \end{table}
        We recall that $\R_3^\text{ex}$ has been defined in \eqref{eq.R3exYpq}. Assuming, without loss of generality, $\imp(\epsilon_1)>\imp(\epsilon_2)>0$, the perturbative partition function \eqref{eq.ZexYpq} becomes:
        \begin{equation}
            Z^\text{ex}_{Y^{p,q}}= \prod_{\alpha\in\Delta}\prod_{t\in\mathbb{Z}}\frac{\Upsilon_1( \ii\alpha(a) +\beta^{-1,\text{ex}}_1t|\epsilon_1^{\text{ex},1},\epsilon_2^{\text{ex},1})\cdot\Upsilon_4( \ii\alpha(a) +\beta^{-1,\text{ex}}_3t|\epsilon_1^{\text{ex},4},\epsilon_2^{\text{ex},4})}{\Upsilon_2( \ii\alpha(a) +\beta^{-1,\text{ex}}_2t|\epsilon_1^{\text{ex},2},\epsilon_2^{\text{ex},2})\cdot\Upsilon_3( \ii\alpha(a) +\beta^{-1,\text{ex}}_4t|\epsilon_1^{\text{ex},3},\epsilon_2^{\text{ex},3})}.
        \end{equation}
        It is now straightforward to obtain $Z_{Y^{p,q}/\mathbb{Z}_h}$ and $Z_{Y^{p,q}/S^1}$ from this expression in the way described above.

    \paragraph{Index Computation.}
        An intrinsically four-dimensional treatment of the $\mathcal{N}=2$ theory on $B$ was proposed in \cite{Festuccia:2018rew}. The one-loop partition function can be obtained via an index computation for the transversally elliptic complex that arises from localisation. This computation is quite involved and has only been performed in the trivial flux sector so far \cite{Mauch:2021fgc}. With the result \eqref{eq--6.Bpert} at hand, at least for SD/ASD distributions reachable from 5d, one can reconstruct the index of the complex, thereby generalising it to non-trivial flux sectors.

\addtocontents{toc}{\protect\setcounter{tocdepth}{1}}

\section{Discussion}\label{sec--7}

    In this work, we have computed the Coulomb branch partition function, including flux contributions, for the 4d $\mathcal{N}=2$ vector multiplet on a large class of closed simply-connected four-manifolds $B$. We started on a  principal $S^1$-bundle over $B$ whose bundle space is a simply-connected toric Sasakian five-manifold. Making use of existing results for this setup, we have computed the one-loop partition function on finite quotients $X=M/\mathbb{Z}_h$. This was done by restricting the modes contributing to the superdeterminant to ones that satisfy the projection condition \eqref{eq--3.18.projcond} for the charge $t$ under the free $S^1$. Most importantly, since $\pi_1(X)\neq0$ we had to include non-trivial flat connections into the localisation locus. Taking the limit of large $h$, we finally obtained the one-loop partition function on $B$ with flux contributions originating from the 5d flat connections. Depending on the relative orientation of $\X$ and $\R$, the theory on $B$ can be either topologically twisted or exotic.
    Finally, we have factorised the one-loop partition functions on $M$, $X$, and $B$ and were able to read out the shifts of the Coulomb branch parameter by the fluxes at each fixed fibre/point. This enabled us to write down the Coulomb branch partition function, including instanton and classical parts. 

    Finally, we point out once again that, whenever $H^2(M;\mathbb{Z})\neq0$, our procedure might only produce part of the sum over all possible fluxes in the full partition function on $B$ (e.g. for $Y^{p,q}$, $H^2(Y^{p,q};\mathbb{Z})\simeq\mathbb{Z}$ and $\mathcal{Z}_{Y^{p,q}}$ is expected to contain a sum over these fluxes). In order to include the remaining flux contributions, it would be necessary to extend the 5d localisation locus accordingly, which we leave for future work. Once this is done, using our procedure, one obtains the full partition function on $B$ (again, for $Y^{p,q}$, $H^2(Y^{p,q}/S^1)\simeq\mathbb{Z}^2$: one sector is carried over from 5d, the other one is introduced by our procedure).

    \paragraph{BPS Strings.}
        Let us suggest an interpretation for our procedure to compute the four-dimensional one-loop partition function around fluxes in terms of the BPS objects of the 5d theories. Restricting, for simplicity, the gauge group to be $SU(2)$, it is known that all rank one 5d SCFTs descend from 6d E-string theory upon circle compactification\footnote{It is conjectured \cite{Jefferson:2018irk} that all 5d $\mathcal{N}=1$ SCFTs can be obtained via an RG flow from 6d $\mathcal{N}=(1,0)$ SCFTs \cite{Witten:1995zh,Strominger:1995ac,Seiberg:1996qx} on a circle}. The degrees of freedom at a generic point in the tensor branch\footnote{All known interacting $\mathcal{N}=(1,0)$ SCFTs include a tensor multiplet, whose components are a tensor field $B$, with anti-self-dual field strength, two fermions, and one scalar. At a generic point of the tensor branch the scalar acquires a vev.} of these 6d $\mathcal{N}=(1,0)$ SCFTs are tensionful strings. These can descend, taking the radius of the sixth dimension to be small, to either electrically charged BPS particles in 5d, if they wrap the $S^1$-fibre, or magnetically charged BPS strings coupling to $F_D=\star F$, if they do not wrap the $S^1$-fibre. If we further take our five-dimensional space-time to be a toric Sasakian manifold $M$, we find a 5d SCFT whose IR description is that of a weakly coupled $SU(2)$ gauge theory on $M$.   
        
        As conjectured in \cite{Lockhart:2012vp}, the partition function of an $\mathcal{N}=1$ vector multiplet on $S^5$ only has contributions from BPS particles while BPS strings do not contribute\footnote{Considering toric Sasakian manifolds with non-trivial $H^2(M;\mathbb{Z})$, it is expected that magnetic strings wrapping two-cycles in $M$ will contribute to flux sectors already in 5d. We limit the discussion in this paragraph to the new flux sectors which arise dimensionally reducing to 4d.}. However, with our dimensional reduction along a non-trivial fibre $S^1\hookrightarrow M\rightarrow B$, we can show explicitly how BPS strings in 5d give rise to new flux sectors in 4d by wrapping the fibre which is the orbit of the free $S^1$-action generated by $\Vec{\X}$. After taking the limit in which the fibre shrinks, these wrapped BPS strings result in magnetically charged particles in the 4d $\mathcal{N}=2$ theories. Therefore, at the trivial sector of these new fluxes arising from dimensional reduction, the partition function $Z_B$ receives contributions only from electrically charged BPS particles. Instead, at a generic flux sector, the partition function receives contributions from both electrically and magnetically charged BPS particles. At the intermediate step, on $X=M/\mathbb{Z}_h$, the partition function at a generic, non-trivial, flat connection sector, receives contributions from magnetically charged BPS strings. Notice that a string wrapping $h$ times the fibre of $X$ does not contribute since it behaves similarly to strings not wrapping the fibre at all. 
        
        A similar discussion holds for the factorised expression presented in \autoref{sec--6}. Thus, one can place the 6d $\mathcal{N}=(1,0)$ SCFT on $\mathbb{C}^2_{\epsilon^i_1,\epsilon^i_2}\times T^2$ and, after dimensionally reducing on the $T^2$, find the contributions of electric and magnetic BPS particles to the one-loop partition function in 4d $\mathcal{N}=2$ at the trivial instanton sector. These contributions arise from 6d effective strings wrapping, respectively, the sixth and the fifth direction.

    \subsection{Future Directions}

        \paragraph{Complete Partition Function.}
            In this work, we have determined the instanton part by computing the shifts of the Coulomb branch parameter at each fixed point and writing the partition function as a product of Nekrasov partition functions around the fixed points. Although this is common practice in 4d, a rigorous proof, to the best of our knowledge, is still missing. 
        
            Moreover, it would be nice to extend our approach to also include hypermultiplets. However, even in the cohomological formulation this requires a spin structure on $M$ and the reduction to 4d has to be carried out more carefully in order for the spin bundle on $M$ to descend to a spin (or at least spin$^c$) bundle on $B$. Such an analysis was performed in \cite{Festuccia:2016gul} for $M$ Sasaki-Einstein (using the ordinary formulation of supersymmetry) and it would be interesting to extend it according to our procedure.

            There are existing results for the partition function of the topologically twisted theory including flux on some compact toric 4-manifolds \cite{Bawane:2014uka,Rodriguez-Gomez:2014eza,Bershtein:2015xfa,Bershtein:2016mxz,Bonelli:2020xps} where the sum is over equivariant fluxes. It would be desirable to compare these to our result in terms of physical fluxes, possibly making use of a suitable notion of S-duality \cite{Festuccia:2020xtv}.  

            Finally, in order to explicitly compute the partition function, the contour integral of the Coulomb branch partition function over $a$ needs to be understood. Since in 5d SYM there is only one scalar field $\sigma$ and the fundamental group is finite, we expect there to be a unique choice of integration contour. If this choice commutes with the $h\to\infty$ limit, then we would have an easy recipe for the 4d contour whenever the 4d theory arises from 5d using our procedure. 
        
        \paragraph{Locally Free $S^1$-Actions.}
            In this work, we assumed the $S^1$-action generated by $\Vec{\X}$ to be globally free such that $B=M/S^1$ has no singularities. More generally, we could consider a locally free $\widetilde{S^1}$-action\footnote{For Sasakian manifolds, such actions are discussed in \cite{Sparks:2010sn}.}. In this case, $M$ will contain points with finite isotropy group that descend to orbifold singularities on $\widetilde{B}=M/\widetilde{S^1}$. The intermediate step would be to consider non-simply-connected orbifolds $\widetilde{X}=M/\widetilde{\mathbb{Z}_h}$, where the quotient acts on the fibre which is the orbit of the $\widetilde{S^1}$-action. 

            Let us briefly sketch how we think the procedure in \autoref{sec--4} can be adapted to these cases. For simplicity, let us consider the toric Sasakian manifold $S^5$ discussed in \autoref{sec--3.1}. We recall that the vectors spanning the edges of the two-dimensional dual cone are:
            \begin{equation}
            \vec{u}_1=[0,1,0],\quad \vec{u}_2=[0,0,1],\quad \vec{u}_3=[1,0,0]
            \end{equation}
            and thus the inward-pointing normals are $\Vec{v}_3=\Vec{u}_2$, $\Vec{v}_2=\Vec{u}_1$, $\Vec{v}_1=\Vec{u}_3$. The reduction of an $\mathcal{N}=1$ vector multiplet is performed along the orbits of the free $S^1$-actions, generated by:
            \begin{equation}
                \Vec{\X}^\text{top}=[1,1,1]\sim\R,\qquad\Vec{\X}^\text{ex}=[1,1,-1].
            \end{equation}
            However, locally free $\widetilde{S^1}$-actions are generated by:
            \begin{equation}\label{eq.nonfreefibres}
                \Vec{\Tilde{\X}}^\text{top}=[l_1,l_2,l_3],\qquad\Vec{\Tilde{\X}}^\text{ex}=[l_1,l_2,-l_3],
            \end{equation}
            where, as long as $l_1,l_2,l_3\in\mathbb{N}$ and $\mbox{gcd}(l_1,l_2,l_3)=1$, the action is effective. Following our procedure, the next step would be to introduce a quotient by $\mathbb{Z}_h$ acting on the fibres \eqref{eq.nonfreefibres}, sum over non-trivial flat connections and eventually take the large $h$ limit. At finite $h$, the resulting space is an orbifold with a conical singularity, representing a deficit angle. An in-depth analysis of this scenario is carried out in \cite{Mauch:2024uyt}. The procedure can be generalized starting from a generic toric Sasakian manifold $M$ and will appear in future work.
            
            The setup above is related to the study of two-dimensional weighted projective spaces, also known as spindles \cite{Ferrero:2020laf}. In \cite{Lundin:2021zeb} the reduction of an $\mathcal{N}=2$ vector multiplet from $S^3$ to $S^2$ is considered. As a consistency check, taking an $\mathcal{N}=1$ vector multiplet $S^3\times S^1$ as starting point and generalizing the reduction to locally free $S^1$-actions, one should reproduce the result of \cite{Inglese:2023wky} for both topological and exotic theories. Moreover, notice that on $S^3$, as the double sine function has huge cancellations between numerator and denominator, we expect the dependence on the choice of locally free $S^1$-action to only affect the shift\footnote{This is already what happens reducing on the two fibres $\vec{\X}^{\text{top}},\vec{\X}^{\text{ex}}$ of $S^3$ \cite{Lundin:2021zeb}.} in $\alpha(\mathfrak{m})$. Instead, reducing from 5d to 4d, we expect the different slicings to depend on the choice of locally free $S^1$-action as the triple sine function does not have such cancellations. These results would significantly enlarge the observables for SQFTs on orbifolds and, in the large-$N$ limit, they would allow an in-depth study of the gravitational block formulas conjectured in \cite{Faedo:2021nub,Faedo:2022rqx}.

        \paragraph{3-Sasakian Reduction}
            Another possibility to extend our procedure is to consider the dimensional reduction to 4d from a 7d $\mathcal{N}=1$ vector multiplet on a 3-Sasakian hypertoric manifold \cite{Polydorou:2017jha,Rocen:2018xwo,Iakovidis:2020znp}. Seven-dimensional 3-Sasakian manifolds are an $S^3$-fibration over a four-dimensional quaternion Kähler orbifold. However, when the $S^3$-action is free, the fibration is over a four-dimensional quaternion Kähler manifold. An example is the Hopf fibration $S^3\hookrightarrow S^7\rightarrow S^4$. For $n=2$, the perturbative partition function is given by a quadruple sine function and it is shown to factorise in contributions arising from the fixed point of the $T^2$-action on the four-dimensional base \cite{Iakovidis:2020znp}. In general, reducing on a sphere breaks half of the supersymmetry and thus, shrinking the size of the $S^3$, one ends up with an $\mathcal{N}=2$ theory in 4d. For example, we would expect to find both an equivariant version of Donaldson-Witten theory and Pestun's theory on $S^4$ from the same seven-dimensional theory. In general, one could consider the reduction to $\mathcal{N}=2$ theories from a generic quaternion Kähler manifold, thus enlarging the class of manifolds considered in our setup. Finally, we have shown how fluxes, whose nature is purely Abelian, arise from dimensional reductions along non-trivial $S^1$-fibrations over $B$. It would be interesting to understand what configurations arise when dimensionally reducing along $SU(2)$-fibrations.

    \paragraph*{Acknowledgments}
    
        We are grateful to Guido Festuccia, Jian Qiu, and Maxim Zabzine for stimulating discussions on the subject. We thank them, Matteo Sacchi, Leonardo Santilli and Itamar Yaakov for comments on the manuscript. We also thank Seidon Alsaody and Thomas Kragh for mathematical advice. RM acknowledges support from the Centre for Interdisciplinary Mathematics at Uppsala University. LR acknowledges support from the Shuimu Tsinghua Scholar Program.

\appendix

\section{Topology of Finite Quotients and the Base}\label{app-topology}

    In this appendix, we collect some topological facts about the quotient manifolds on which we place the $\mathcal{N}=1$ SYM theory in the main part of this work and the four-dimensional base.

    \subsection{Finite Quotients}
    
        Let $M$ be a closed, simply-connected manifold that has a free $S^1$-action. We consider the quotient of $M$ by the discrete subgroup $\mathbb{Z}_h\subset S^1$, $h\in\mathbb{N}_{\ge2}$. This quotient $X=M/\mathbb{Z}_h$ is again a smooth manifold with a free $(S^1/\mathbb{Z}_h\simeq S^1)$-action . However, $X$ is not simply-connected but has $\pi_1(X)\simeq\mathbb{Z}_h$. This follows from the fact that there is a fibration
        \begin{equation}
            \begin{tikzcd}[column sep=small,row sep=small]
                \mathbb{Z}_h\ar[r] & M\ar[d]\\
                & X
            \end{tikzcd}
        \end{equation}
        giving rise to a long exact sequence (LES) in homotopy:
        \begin{equation}\label{eq--LES.homotopy}
            \begin{tikzcd}[column sep=small,row sep=tiny]
                \dots\ar[r] & \pi_1(M)\ar[d,equal]\ar[r] & \pi_1(X)\ar[r] & \pi_0(\mathbb{Z}_h)\ar[d,equal]\ar[r] & \pi_0(M)\ar[d,equal]\ar[r] & \dots\\
                & 0 & & \mathbb{Z}_h & 0 &
            \end{tikzcd}
        \end{equation}
        Moreover, by continuing the LES above, we find $\pi_i(M)\simeq\pi_i(X)$ for $i>2$.
        
        In order to make sense of the projection condition \eqref{eq--3.2.projcond} we first use the following\footnote{In the following, we always consider (co)homology with integral coefficients, unless specified otherwise.}
            
        \begin{fact}\label{fact.1}
            $H^2(X)\simeq\mathbb{Z}_h\oplus\mathbb{Z}^{b_2(X)}$, where $b_2(X)$ denotes the second Betti number of $X$.
        \end{fact}
        \begin{proof}
            We already know $\pi_1(X)\simeq\mathbb{Z}_h$ and thus, by the Hurewicz theorem, $H_1(X)\simeq\mathbb{Z}_h$. Now we apply the universal coefficient theorem (UCT) for integral coefficients:
            \begin{equation}\label{eq--UCT}
                \begin{tikzcd}[column sep=small]
                    0\ar[r] & \Ext^1_\mathbb{Z}(H_1(X),\mathbb{Z})\ar[r] & H^2(X)\ar[r] & \Hom_\mathbb{Z}(H_2(X),\mathbb{Z})\ar[r] & 0
                \end{tikzcd}.
            \end{equation}
            Since $H_\bullet(X)$ is finitely generated, we can decompose $H_\bullet(X)\simeq\mathbb{Z}^{b_\bullet(X)}\oplus T_\bullet$ with $b_\bullet(X)$ the Betti number and $T_\bullet$ the torsion part. But (see e.g. \cite{Hatcher:2002} p.195):
            \begin{align}
                \Hom_\mathbb{Z}(H_\bullet(X),\mathbb{Z})&\simeq\Hom_\mathbb{Z}(\mathbb{Z}^{b_\bullet(X)},\mathbb{Z})\oplus\Hom_\mathbb{Z}(T_\bullet,\mathbb{Z})\simeq\mathbb{Z}^{b_\bullet(X)},\\
                \Ext^1_\mathbb{Z}(H_\bullet(X),\mathbb{Z})&\simeq\Ext^1_\mathbb{Z}(\mathbb{Z}^{b_\bullet(X)},\mathbb{Z})\oplus\Ext^1_\mathbb{Z}(T_\bullet,\mathbb{Z})\simeq T_\bullet,\label{eq--ext.1}
            \end{align}
            where we have used that $\Hom$ and $\Ext$ preserve limits. Since \eqref{eq--UCT} splits (although not naturally), we arrive at the result.
        \end{proof}
        
        \begin{eg}
            Consider the space $Y^{p,q}$ from \autoref{subsec.Ypq} with $\gcd(p,q)\ge1$. Then $\pi_1(Y^{p,q})\simeq\mathbb{Z}_{\gcd(p,q)}$ (see \cite{Gauntlett:2004yd}, appendix A) and $H^2(Y^{p,q})\simeq\mathbb{Z}_{\gcd(p,q)}\oplus\mathbb{Z}$. 
        \end{eg}
        
        \autoref{fact.1} shows that the torsion part of $H^2(X)$ originates from $\pi_1(X)$. Since flat line bundles $\mathbb{C}^\times\hookrightarrow L\rightarrow X$ are characterised by the first Chern class $c_1(L)$, we need $c_1(L)\in\mathbb{Z}_h\subset H^2(X)$ for the projection condition \eqref{eq--3.2.projcond} to make sense. This is guaranteed by the following
        
        \begin{fact}
            $\im c_1\simeq\Ext^1(\pi_1(X),\mathbb{Z})$.
        \end{fact}
        \begin{proof}
            First, note that (isomorphism classes of) flat line bundles are precisely\footnote{We have $H^1(X;\mathbb{C}^\times)\simeq H^1(X;S^1)\simeq\Hom(\pi_1(X),S^1)$, where $X$ connected.} elements of $H^1(X;{\mathbb{C}}^\times)$. We have a SES of coefficient rings
            \begin{equation}
                \begin{tikzcd}[column sep=scriptsize]
                    0\ar[r] & {\mathbb{Z}}\ar[r] & {\mathbb{C}}\ar[r,"\exp"] & {\mathbb{C}}^\times\ar[r] & 0
                \end{tikzcd},
            \end{equation}
            giving rise to a LES in cohomology:
            \begin{equation}
                \begin{tikzcd}[column sep=scriptsize]
                    \dots\ar[r] & H^1(X;{\mathbb{C}})\ar[r,"f"] & H^1(X;{\mathbb{C}}^\times)\ar[r,"c_1"] & H^2(X;{\mathbb{Z}})\ar[r] & \dots
                \end{tikzcd}
            \end{equation}
            with the connecting homomorphism being the first Chern class (\cite{Putman:2012}, sec. 2.2).
            Since $X$ is connected, the universal coefficient theorem implies
            \begin{equation}
                H^1(X;{\mathbb{C}})\simeq\Hom(H_1(X;{\mathbb{Z}}),{\mathbb{C}}),\qquad H^1(X;{\mathbb{C}}^\times)\simeq\Hom(H_1(X;{\mathbb{Z}}),{\mathbb{C}}^\times).
            \end{equation}
            Under the isomorphisms above, the map $f$ in the LES corresponds to the map
            \begin{equation}
                \tilde{f}:\Hom(H_1(X;{\mathbb{Z}}),{\mathbb{C}})\longrightarrow\Hom(H_1(X;{\mathbb{Z}}),{\mathbb{C}}^\times),\;\varphi\longmapsto\exp\circ\varphi.
            \end{equation}
            The maps in the codomain that cannot be reached by $\tilde{f}$ are precisely those that map torsion elements non-trivially, hence $\coker\tilde{f}\simeq\Hom(H_1(X;{\mathbb{Z}})^\mathrm{tor},{\mathbb{C}}^\times)$. Using the first isomorphism theorem and the fact that $H_1(X;\mathbb{Z})^\mathrm{tor}$ is a finite Abelian group, we obtain
            \begin{equation}
                \im c_1\simeq\Hom(H_1(X;{\mathbb{Z}})^\mathrm{tor},{\mathbb{C}}^\times)\simeq\Ext^1(H_1(X;\mathbb{Z}),\mathbb{Z}).
            \end{equation}
            Applying the Hurewicz theorem gives the result.
        \end{proof}
    
        Hence, $c_1$ takes values in $\pi_1(X)\simeq\mathbb{Z}_h$ and the projection condition \eqref{eq--3.18.projcond} is justified.

    \subsection{The Base}

        As for the topology of $B$, we have a fibration (in particular, a principal $S^1$-bundle) over the base space $B=M/S^1$,
        \begin{equation}
            \begin{tikzcd}[column sep=small,row sep=small]
                S^1\ar[r] & M\ar[d]\\
                & B
            \end{tikzcd}
        \end{equation}
        from which we obtain the LES in homotopy:
        \begin{equation}
            \begin{tikzcd}[column sep=small,row sep=tiny]
                \dots\ar[r] & \pi_2(S^1)\ar[r]\ar[d,equal] & \pi_2(M)\ar[r]\ar[d,equal] & \pi_2(B)\ar[r,"f"] & \pi_1(S^1)\ar[r]\ar[d,equal] & \pi_1(M)\ar[r]\ar[d,equal] & \pi_1(B)\ar[r] & \pi_0(S^1)\ar[r]\ar[d,equal] & \dots \\
                & 0 & H^3(M) & & \mathbb{Z} & 0 & & 0 &
            \end{tikzcd}
        \end{equation}
        We conclude that $B$ is again simply-connected and (using the Hurewicz theorem and Poincaré duality) we have $\mathbb{Z}\subset H^2(B)$ from the surjection $f$. This is the subgroup generated by the two-cycle that is obtained from the three-cycle in $X$; the image of $c_1$ upon reduction takes values in this subgroup. However, the corresponding domain of $c_1$ are no longer flat connections ($\pi_1(B)=0$ and thus $H^1(B;\mathbb{C}^\times)=0$) but the subspace of connections $A$ with curvature $F_A$ satisfying \eqref{eq--3.3.flux} and $c_1(F_A)\in\mathbb{Z}\subset H^2(B)$.  

        Finally, we present a more formal argument for the dimensional reduction of $M$ to $B=M/S^1$ via the quotient $X=M/\mathbb{Z}_h$. Let $X_n:=M/Z_n$ and $Z_n:=\mathbb{Z}/h^n\mathbb{Z}$. Then we have the direct system
        \begin{equation}\label{eq--app.directsystem}
            \begin{tikzcd}[]
                X_1\ar[r,two heads,"f_1"] & X_2\ar[r,two heads,"f_2"] & X_3\ar[r,two heads,"f_3"] & X_4\ar[r,two heads,"f_4"] &\dots
            \end{tikzcd}
        \end{equation}
        with $f_n$ the canonical projections. The colimit of \eqref{eq--app.directsystem} (in the category of compactly generated weakly Hausdorff spaces, CGWH) is given by
        \begin{equation}
            \begin{tikzcd}
                X_1\ar[r,two heads,"f_1"]\ar[drr,"\iota_1"'] & X_2\ar[r,two heads,"f_2"]\ar[dr,"\iota_2"] & X_3\ar[r,two heads,"f_3"]\ar[d,"\iota_3"] & X_4\ar[r,two heads,"f_4"]\ar[dl,"\iota_4"] &\dots\\
                & & \varinjlim X_n & &
            \end{tikzcd}
        \end{equation}
        such that the triangles commute.
        On the other hand, by virtue of the fibrations 
        \begin{equation}
            \begin{tikzcd}[column sep=small]
                S^1\ar[r,hook] & X_n\ar[r,"\pi_n"] & B,
            \end{tikzcd}
        \end{equation}
        we have another co-cone for \eqref{eq--app.directsystem} given by
        \begin{equation}
            \begin{tikzcd}
                X_1\ar[r,two heads,"f_1"]\ar[drr,"\pi_1"'] & X_2\ar[r,two heads,"f_2"]\ar[dr,"\pi_2"] & X_3\ar[r,two heads,"f_3"]\ar[d,"\pi_3"] & X_4\ar[r,two heads,"f_4"]\ar[dl,"\pi_4"] &\dots\\
                & & B & &
            \end{tikzcd}
        \end{equation}
        Then, by the universal property of $\varinjlim X_n$, there is a unique, continuous map $\phi:\varinjlim X_n\rightarrow B$ such that $\pi_n=\phi\circ\iota_n$ for all $n\in\mathbb{N}$. We now want to find a map $\psi$ such that the triangles in the following diagram commute:
        \begin{equation}
            \begin{tikzcd}
                \dots\ar[r,two heads] & X_n\ar[rr,two heads,"f_n"]\ar[dr,"\pi_n"]\ar[ddr,"\iota_n"'] & & X_{n+1}\ar[r,two heads]\ar[dl,"\pi_{n+1}"']\ar[ddl,"\iota_{n+1}"] & \dots\\
                & & B\ar[d,"\psi"] & &\\
                & & \varinjlim X_n & &
            \end{tikzcd}
        \end{equation}
        There is a unique such map $\psi:B\rightarrow\varinjlim X_n$, $b\mapsto\iota_1(\tilde{b})$ for some $\tilde{b}\in X_1$ such that $\pi_1(\tilde{b})=b$. But then $\psi$ must be the unique isomorphism. Moreover, since both spaces are compact and Hausdorff\footnote{This essentially follows from the fact that the weak Hausdorffification of $\varinjlim S^1/Z_n$ in Top is a point.}, $\psi$ is a homeomorphism. In fact, the smooth structure on $\varinjlim X_n$ expected from the quotienting procedure is precisely the one induced from $B$, which turns $\psi$ into a diffeomorphism.

\addtocontents{toc}{\protect\setcounter{tocdepth}{2}}

\bibliographystyle{utphys}
\bibliography{main}

\end{document}